\documentclass{llncs}

\usepackage{graphicx}
\usepackage{amsmath}
\usepackage{amssymb}
\usepackage{color}
\usepackage{url}

\newcommand{\define}{\widehat{=}}
\newcommand{\xx}{\mathbf{x}}

\newcommand{\uu}{\mathbf{u}}

\newcommand{\oomit}[1]{}

\begin{document}

%% no pagenumber
%%\pagestyle{plain}

\title{A ``Hybrid" Approach for Synthesizing Optimal Controllers of Hybrid Systems: A Case Study of the Oil Pump Industrial Example}

%% running head
\titlerunning{A ``Hybrid" Approach for Synthesizing Optimal Controllers}

\author % for multiple authors
{Hengjun Zhao\inst{1,2} \and Naijun Zhan\inst{2} \and Deepak Kapur\inst{3}\and Kim G. Larsen\inst{4}\,
%\thanks{The first and second authors are supported by NSFC projects 91118007 and 60970031; the third author is supported by NSF CCF-0729097 and CNS-0905222; the fourth author is supported by The Danish VKR Center of Excellence MT-LAB and
%The Sino-Danish Basic Research Center IDEA4CPS.}
}
\institute
{
  Graduate University of Chinese Academy of Sciences, Beijing, China
  \and
  State Key Lab. of Comput. Sci.,
  Institute of Software, CAS, Beijing, China\\
 %% \email{\{znj,zhaohj\}@ios.ac.cn}
  \and
  Dept.~of Comput. Sci., University of New Mexico,
  Albuquerque, NM, USA\\
 %% \email{kapur@cs.unm.edu}
  \and CISS, CS, Aalborg University, Denmark
}
\authorrunning{H. Zhao, N. Zhan, D. Kapur and K.G. Larsen}

\maketitle

\begin{abstract}
In this paper, we propose an approach to reduce the optimal controller synthesis problem of hybrid systems to quantifier elimination; furthermore, we also show how to combine quantifier elimination with numerical computation in order to make it more scalable but at the same time, keep arising errors due to discretization manageable and within bounds.  A major advantage of our approach is not only that it avoids errors due to numerical computation, but it also gives a better optimal controller. In order to illustrate our approach, we use the real industrial example of an oil pump provided by the German company HYDAC within the European project \emph{Quasimodo} as a case study throughout this paper, and show that our method improves (up to $7.5\%$) the results reported in \cite{kim09} based on game theory and model checking.
\end{abstract}


\section{Introduction}
Hybrid systems such as physical devices controlled by computer software, are systems that exhibit both continuous and discrete behaviors. Controller synthesis for hybrid systems is an important area of research in both academia and industry. A synthesis problem focuses on designing a controller that ensures the given system will satisfy a safety requirement, a liveness requirement (e.g. reachability to a given set of states), or meet an optimality criterion, or a desired combination of these requirements.

Numerous work have been done on controller synthesis for safety and/or reachability requirements.
For example, in \cite{Pnueli00,Sastry00}, a
general framework for synthesizing controllers based on hybrid
automata to meet a given safety requirement was proposed, which relies on \emph{backward reachable set} computation and \emph{fixed point iteration};
while in \cite{Taly10}, a symbolic approach based on templates and constraint solving to the same problem was proposed, and in \cite{Taly-ctrl-10}, the symbolic approach is extended to meet both safety and reachability requirements.

However, the optimal controller synthesis problem is more involved, also quite important in the design of hybrid systems. In the literature, few work has been done on the problem. Larsen et al proposed an approach based on energy automata and model-checking \cite{kim09}, while Jha, Seshia and Tiwari gave a solution to the problem using unconstrained numerical optimization and machine learning \cite{Tiwari-emsoft11}.
However, in \cite{kim09}, allowing control only to be exercised at discrete points in time certainly limits the opportunity of synthesizing the optimal controller (though one can get arbitrarily close).  Moreover, discretizing  could cause an incorrect controller to be synthesized --- which therefore requires a posterior analysis (e.g. in [3], PHAV{\small{ER}} \cite{Phaver} is used for the purpose).
The approach of \cite{Tiwari-emsoft11} suffers from imprecision caused by numerical computation, and  cannot synthesize a really optimal controller sometimes because the machine learning technique cannot guarantee
its completeness.

In this paper, we propose a ``hybrid" approach for synthesizing optimal controllers of hybrid systems subject to safety requirements. The basic idea is as follows. Firstly, we reduce optimal controller synthesis subject to safety requirements to quantifier elimination (QE for short). Secondly, in order to make our approach scalable, we discuss how to combine QE with numerical computation, but at the same time, keep arising errors due to discretization manageable and within bounds. A major advantage of our approach is not only that it avoids errors due to numerical computation, but it also gives a better optimal controller.

Application of QE in controller synthesis of hybrid systems is not new.
The tool HyTech was the first symbolic model checker that can do parametric analysis \cite{Hytech} for linear hybrid automata, but for the oil pump example it will abort soon due to arithmetic overflow errors. Recently, verification and synthesis of switched dynamical systems using QE were discussed in \cite{sturm-tiwari-issac11}, where the authors gave principles and heuristics for combining different tools, to solve QE problems that are out of the scope of each component tool.

Our encoding of a {\small{MIN-MAX-MIN}} optimization problem into a QE problem is inspired by the idea in \cite{Dolzmann98}: minimizing an objective function $f(x_1,x_2,\ldots,x_n)$ can be solved by introducing an additional constraint $z\geq f(x_1,x_2,\ldots,x_n)$ and eliminating variables $x_1,x_2,\ldots,x_n$, where $z$ is a newly introduced variable. Similar ideas can also be found in \cite{Rupak10}.

The computation of optimal control strategies in this paper is typically a \emph{parametric optimization} problem, a topic researched extensively in both operation research and control communities. Symbolic methods have advantages in addressing parametric optimization problems \cite{Weispfenning94,Parrilo06,Kanno08}. However, we do not find any algorithm suitable for solving a parametric quadratic optimization problem over constraint with complex Boolean structure and hundreds of (or thousands of) atomic formulas as in this paper.

It was shown in \cite{Bemporad02} that for certain parametric quadratic optimization problems, the closed form solution exists: the optimizer is a piecewise affine function in the parameters, and the optimal value is a piecewise quadratic function in the parameters. Our experiment results confirm this.

In order to illustrate our approach, we use the oil pump industrial example provided by the German company {\small{HYDAC}} within the European project \emph{Quasimodo} as a case study throughout this paper, and show that our method results in a better optimal controller (up to $7.5\%$ improvement) than those reported in \cite{kim09} based on game theory and model checking. Moreover, we prove that the theoretically optimal controller of the oil pump example can be synthesized and its correctness is also guaranteed with our approach.

\subsubsection{Paper Organization:}
In Section 2 we propose a general framework for optimal controller synthesis of hybrid systems based on quantifier elimination and numerical computation. We focus on the oil pump case study in Section 3-6: a description of the oil pump control problem is given in Section 3, modeling of the system and safety requirements is shown in Section 4, a ``hybrid" approach for performing optimization is presented in Section 5, and further improvement through a modification in the model is discussed in Section 6. We finally conclude this paper by Section 7.

\section{The Overall Approach}\label{sec:overall}
In this section we propose an approach that reduces optimal controller synthesis of hybrid systems subject to safety requirements to QE. Such reduction is based on reachable set computation or approximation of hybrid systems and symbolic optimization.
We also discuss how numerical computation can be incorporated into our approach to make it more scalable.

Generally, a hybrid system consists of a set of continuous state variables $\xx$ (ranging over $\mathbb R^n$) and a set of discrete operating modes $Q$, with each of which a continuous dynamics is associated specifying the behavior of $\xx$ at each mode; discrete jumps between different modes may happen if some \emph{transition conditions} are satisfied by $\xx$.

The optimal controller synthesis problem studied in this paper can be stated as follows. Suppose we are given an under-specified hybrid system $\mathcal H$, in which the transition conditions are not determined but parameterized by $\uu$, a vector of control parameters. Our task is to determine values of $\uu$ such that $\mathcal H$ can make discrete jumps at desired points, thus guaranteeing that
\begin{itemize}
  \item[1)] a safety requirement $\mathcal S$ is satisfied, that is, $\xx$ stays in a designated safe region at any time point; and
  \item[2)] an optimization goal $\mathcal G$, possibly
$$\min_{\uu}g(\uu),\,\,\max_{\uu_2}\min_{\uu_1}g(\uu)\,,\, \mathrm{or}\,\,\min_{\uu_3}\max_{\uu_2}\min_{\uu_1}g(\uu)\,,\setcounter{footnote}{0}
\footnote{We assume that $\uu$ is chosen from a compact set, and elements of $\uu$ are divided into groups $\uu_1,\uu_2,\uu_3,\ldots$ according to their roles in $\mathcal G$.}
$$
where $g(\uu)$ is an objective function in parameters $\uu$,
is achieved.
\end{itemize}
Then our approach for solving the synthesis problem can be described as the following steps.
\subsubsection{Step 1.} {\emph{Derive constraint $D(\uu)$ on $\uu$ from safety requirements of the system.}}

If the {reachable set} $R$ (parameterized by $\uu$) of $\mathcal H$ can be exactly computed (e.g. for very simple linear hybrid automata), then we just require that $R$ should be contained in the safe region. Otherwise we have to approximate $R$ (with sufficient precision) by automatically generating inductive invariants of $\mathcal H$ (e.g. for general linear or nonlinear hybrid systems). The notion of \emph{inductive invariant} is crucial in safety verification of hybrid systems \cite{Tiwari08,PlatzerClarke09}, and constraint-based approaches have been  proposed for automatic generation of inductive invariants \cite{Manna04,Tiwari08,PlatzerClarke08,emsoft11}.

\subsubsection{Step 2.} {\emph{Encode the optimization problem $\mathcal G$ over constraint $D(\uu)$ into a quantified first-order formula $\mathbf Q\mathbf u.\varphi(\uu,z)$, where $z$ is a fresh variable.}}

Our encoding is based on the following proposition.
\begin{proposition}\label{prop:optimization}
Suppose $g_1(\uu_1)$, $g_2(\uu_1,\uu_2)$, $g_3(\uu_1,\uu_2,\uu_3)$ are polynomials, and $D_1(\uu_1)$, $D_2(\uu_1,\uu_2)$, $D_3(\uu_1,\uu_2,\uu_3)$ are nonempty compact semi-algebraic sets\footnote{A semi-algebraic set is defined by Boolean combinations of polynomial equations and inequalities.}.
Then there exist $c_1,\,c_2,\,c_3\in \mathbb R$ s.t.
\begin{eqnarray}
  & \exists \uu_1.(D_1\wedge g_1\leq z)\,\Longleftrightarrow\, z\geq c_1\,,\label{eqn:minz}\\
  & \forall \uu_2.\big(\exists \uu_1. D_2\longrightarrow\exists \uu_1.(D_2\wedge g_2\leq z)\big)\,\Longleftrightarrow\, z\geq c_2\,,\label{eqn:supminz}\\
  & \exists \uu_3.\big(
    (\exists \uu_1\uu_2. D_3) \,\wedge\,\forall \uu_2.\big(\exists \uu_1. D_3\longrightarrow\exists \uu_1.(D_3\wedge g_3\leq z)\big)\big)\,\Longleftrightarrow\, z\vartriangleright c_3 \,, \label{eqn:infsupminz}
\end{eqnarray}
where $\vartriangleright\in\{>,\geq\}$, and $c_1,c_2,c_3$ satisfy
\begin{eqnarray}
      c_1&\,=\,&\min_{\uu_1} g_1(\uu_1) \,\quad\mathrm{over}\,D_1(\uu_1)\,,\label{eqn:c1}\\
      c_2&\,=\,& \underset{\,\uu_2\,\,\,\,\,\,\uu_1}{\sup\min}\, g_2(\uu_1,\uu_2) \,
      \quad\mathrm{over}\,D_2(\uu_1,\uu_2)\,,\label{eqn:c2}\\
      c_3&\,=\,&\underset{\uu_3\,\,\,\,\,\,\uu_2\,\,\,\,\,\,\uu_1}{\inf\sup\min} \, g_3(\uu_1,\uu_2,\uu_3) \,\quad\mathrm{over}\,D_3(\uu_1,\uu_2,\uu_3)\,.
    \label{eqn:c3}
\end{eqnarray}
\end{proposition}

\begin{proof}
  Given assumptions in Proposition \ref{prop:optimization}, the following facts are easy to check:
  \begin{itemize}
    \item[(f1)] $\exists \uu_1. D_2(\uu_1,\uu_2)$ is a compact set over $\uu_2$;
    \item[(f2)] for any $\uu_2^{*}$ satisfying $\exists \uu_1. D_2(\uu_1,\uu_2)$, the instantiation of $D_2$ by $\uu_2^*$, i.e. $D_2(\uu_1,\uu_2^{*})$ is a compact set over $\uu_1$;
    \item[(f3)] results similar to (f1) and (f2) can be established for $D_3(\uu_1,\uu_2,\uu_3)$\,.
  \end{itemize}

  First we show the existence of $c_1,c_2,c_3$ in (\ref{eqn:c1}), (\ref{eqn:c2}) and (\ref{eqn:c3}).

  \emph{Proof of (\ref{eqn:c1}):} The existence of $c_1$ is based on the \emph{Extreme Value Theorem}: a real-valued continuous function has a \emph{minimum} and a \emph{maximum} on a compact set.

  \emph{Proof of (\ref{eqn:c2}):} Let $$\overline{c_2}=\max_{\uu_1,\uu_2}g_2(\uu_1,\uu_2)\quad\mathrm{over}\quad D_2(\uu_1,\uu_2)\,.$$
  Then for any $c_2^*$ satisfying $\exists \uu_1.D_1$
  $$\min_{\uu_1}g_2(\uu_1,\uu_2^*)\quad\mathrm{over}\quad D_2(\uu_1,\uu_2^*)$$ exists and
  $$\min_{\uu_1}g_2(\uu_1,\uu_2^*)\leq \overline{c_2}\,.$$
  Therefore the supremum of $\min_{\uu_1}g_2(\uu_1,\uu_2)$ over $D_2$, i.e. $c_2$, exists.

  \emph{Proof of (\ref{eqn:c3}):} Let $$\underline{c_3}=\min_{\uu_1,\uu_2,\uu_3}g_3(\uu_1,\uu_2,\uu_3)\quad\mathrm{over}\quad D_3\,.$$
  Then
  $$\sup_{\uu_2}\min_{\overset{}{\uu_1}}g_3(\uu_1,\uu_2,\uu_3)\quad\mathrm{over}\quad D_3$$ is lower bounded by $\underline{c_3}$. Thus $c_3$ exists.

  Next we will prove (\ref{eqn:minz}) -- (\ref{eqn:infsupminz}). For brevity, in the sequel we use ($\cdot$)$^l$ and ($\cdot$)$^r$ to denote the left and right hand side sub-formulas in the equivalence relations (\ref{eqn:minz}) -- (\ref{eqn:infsupminz}).

  \emph{Proof of (\ref{eqn:minz}):} ``$\Rightarrow$" Suppose $z$ satisfies (\ref{eqn:minz})$^l$ but $z<c_1$. Then there exists $\uu_1^*\in D_1$ s.t. $$c_1>z\geq g_1(\uu_1^*),$$
  which contradicts (\ref{eqn:c1}); ``$\Leftarrow$" Suppose $z$ satisfies (\ref{eqn:minz})$^r$. By (\ref{eqn:c1}) we have $c_1=g_1(\uu_1^*)$ for some $\uu_1^*\in D_1$. Thus
  $$z\geq c_1=g_1(\uu_1^*)\,,$$
  so $z$ satisfies (\ref{eqn:minz})$^l$.

  \emph{Proof of (\ref{eqn:supminz}):} ``$\Rightarrow$" Suppose $z$ satisfies (\ref{eqn:supminz})$^l$. Then for all $\uu_2^*$ in $\exists \uu_1. D_2$ we have
  $$\exists \uu_1.\big(D_2(\uu_1,\uu_2^*)\wedge g_2(\uu_1,\uu_2^*)\leq z\big)\,.$$
  By (\ref{eqn:minz}) it follows that
  $$z\geq \min_{\uu_1}g_2(\uu_1,\uu_2^*)\quad\mathrm{over}\quad D_2(\uu_1,\uu_2^*)$$
  for all $\uu_2^*$ in $\exists \uu_1. D_2$,  so by (\ref{eqn:c2}) $z\geq c_2\,.$
  ``$\Leftarrow$"  Suppose $z$ satisfies (\ref{eqn:supminz})$^r$. Then by (\ref{eqn:c2}) we have for all $\uu_2^*$ in $\exists \uu_1. D_2$
  $$z\geq \min_{\uu_1}g_2(\uu_1,\uu_2^*)\quad\mathrm{over}\quad D_2(\uu_1,\uu_2^*)\,.  $$ Again by (\ref{eqn:minz}) we get
  $$\exists \uu_1.\big(D_2(\uu_1,\uu_2^*)\wedge g_2(\uu_1,\uu_2^*)\leq z\big)$$
  holds for all $\uu_2^*$ in $\exists \uu_1. D_2$, which means $z$ satisfies
  (\ref{eqn:supminz})$^l$\,.

  \emph{Proof of (\ref{eqn:infsupminz}):}
  The proof below is based on the fact that if \emph{infimum} in (\ref{eqn:c3}) is actually \emph{minimum}, then (\ref{eqn:infsupminz})$^r$ is $z\geq c_3$; otherwise (\ref{eqn:infsupminz})$^r$ is $z>c_3$. We only give the proof for the former case.

  ``$\Rightarrow$" Suppose $z$ satisfies (\ref{eqn:infsupminz})$^l$. Then there exists $\uu_3^*$ in $\exists \uu_1\uu_2.D_3$ s.t.
  $$\forall \uu_2.\big(\exists \uu_1.D_3(\uu_1,\uu_2,\uu_3^*)\longrightarrow \exists \uu_1. (D_3(\uu_1,\uu_2,\uu_3^*)\wedge g_3(\uu_1,\uu_2,\uu_3^*)\leq z)\big)\,.$$
  By (\ref{eqn:supminz}) we have
  $$z\geq \sup_{\uu_2}\min_{\overset{}{\uu_1}}g_3(\uu_1,\uu_2,\uu_3^*)\quad \mathrm{over}\quad D(\uu_1,\uu_2,\uu_3^*)\,.$$
  Then by (\ref{eqn:c3}) $z$ satisfies (\ref{eqn:infsupminz})$^r$. ``$\Leftarrow$" Suppose $z$ satisfies (\ref{eqn:infsupminz})$^r$. Then by
  (\ref{eqn:c3}) we assert that there exists $\uu_3^*$ in $\exists\uu_1\uu_2.D_3$ s.t.
  $$z\geq \sup_{\uu_2}\min_{\overset{}{\uu_1}}g_3(\uu_1,\uu_2,\uu_3^*)\quad \mathrm{over}\quad D(\uu_1,\uu_2,\uu_3^*)\,.$$
  Again by (\ref{eqn:supminz}) it follows that $z$ satisfies
  $$\forall \uu_2.\big(\exists \uu_1.D_3(\uu_1,\uu_2,\uu_3^*)\longrightarrow \exists \uu_1. (D_3(\uu_1,\uu_2,\uu_3^*)\wedge g_3(\uu_1,\uu_2,\uu_3^*)\leq z)\big)\,.$$
  Thus $z$ satisfies (\ref{eqn:infsupminz})$^l$.

  If infimum in (\ref{eqn:c3}) is not minimum, an analogous proof can be given.
  \qed
\end{proof}

%We omit the proof of this proposition due to limit of space.
%All the proof details in this paper, as well as the formulas by QE,
%can be found in the full version \cite{fmfull}.
%Please refer to Appendix for the proof of this proposition.
%\begin{remark}
%If we do quantifier elimination on the left-hand side formulas in (\ref{eqn:minz}) -- (\ref{eqn:minmaxminz}), the results returned by tools  may have very complicated forms, but they are still equivalent to $z\geq c_1,\, z\geq c_2, z\geq c_3$ or $z> c_3$\,.
%\end{remark}

\subsubsection{Step 3.} {\emph{Eliminate quantifiers in $\mathbf Q\mathbf u.\varphi(\uu,z)$ and from the result we can retrieve the optimal value of $\mathcal G$ and the corresponding optimal controller $\uu$.}}

By Proposition \ref{prop:optimization}, the optimal value of a {\small{MIN}}, {\small{MAX-MIN}} or {\small{MIN-MAX-MIN}} problem can be obtained by applying QE to the left hand side (LHS) formulas in (\ref{eqn:minz})-(\ref{eqn:infsupminz}) respectively.
Although QE for the first-order theory of real closed fields is a complete decision procedure \cite{Tarski51}, due to the inherent doubly exponential complexity \cite{dh88}, we cannot expect to compute an optimal value, say $c_3$, by directly applying QE to a big formula with many alternations of quantifiers, like LHS of (\ref{eqn:infsupminz}).
Therefore it is necessary to devise our own mechanisms for performing QE more efficiently.

Note that in (\ref{eqn:infsupminz}), any instantiation of the outmost quantified variables $\uu_3$ would result in a simpler formula, whose quantifier-free equivalence gives an upper bound of $c_3$. If in some way we know the bounds of $\uu_3$, i.e. $l_i\leq \uu_3^i\leq u_i$, for $1\leq i\leq \mathrm{dim}(\uu_3)$, then by discretizing $\uu_3$ over all $[l_i,u_i]$ with certain granularity $\Delta$, and using the set of discretized values to instantiate the outmost existential quantifiers of (\ref{eqn:infsupminz}), we can get a finite set of simplified formulas, each of which produces an upper approximation of $c_3$. Finally, through an exhaustive search in this set we can select such an approximation that is closest to $c_3$. Finer granularity yields better approximation of the optimal value, so one can seek for a good balance between timing and optimality by tuning the granularity $\Delta$. Furthermore, the above computation is well suited for parallelization to make full use of available computing resources, because the intervals $[l_i,u_i]$ and corresponding instantiations can be divided into subgroups and allocated to different processes.

\section{Description of the Oil Pump Control Problem}\label{sec:oil}
The oil pump example \cite{kim09} was a real industrial case provided by the German company H{\small{YDAC}} E{\small{LECTRONICS}} G{\small{MBH}}, and studied at length within the European research project {\emph{Quasimodo}}.
The whole system, depicted by Fig.~\ref{fig:sys}, consists of a machine, an accumulator, a reservoir and a pump. The machine consumes oil periodically out of the accumulator with a duration of $20\,s$ (second) for one consumption cycle. The profile of consumption rate is shown in Fig.~\ref{fig:consump}. The pump adds oil from the reservoir into the accumulator with power $2.2\,l/s$ (liter/second).
\begin{figure}
\begin{minipage}[t]{.45\textwidth}
\begin{center}
\includegraphics[width=1.3in,height=1in]{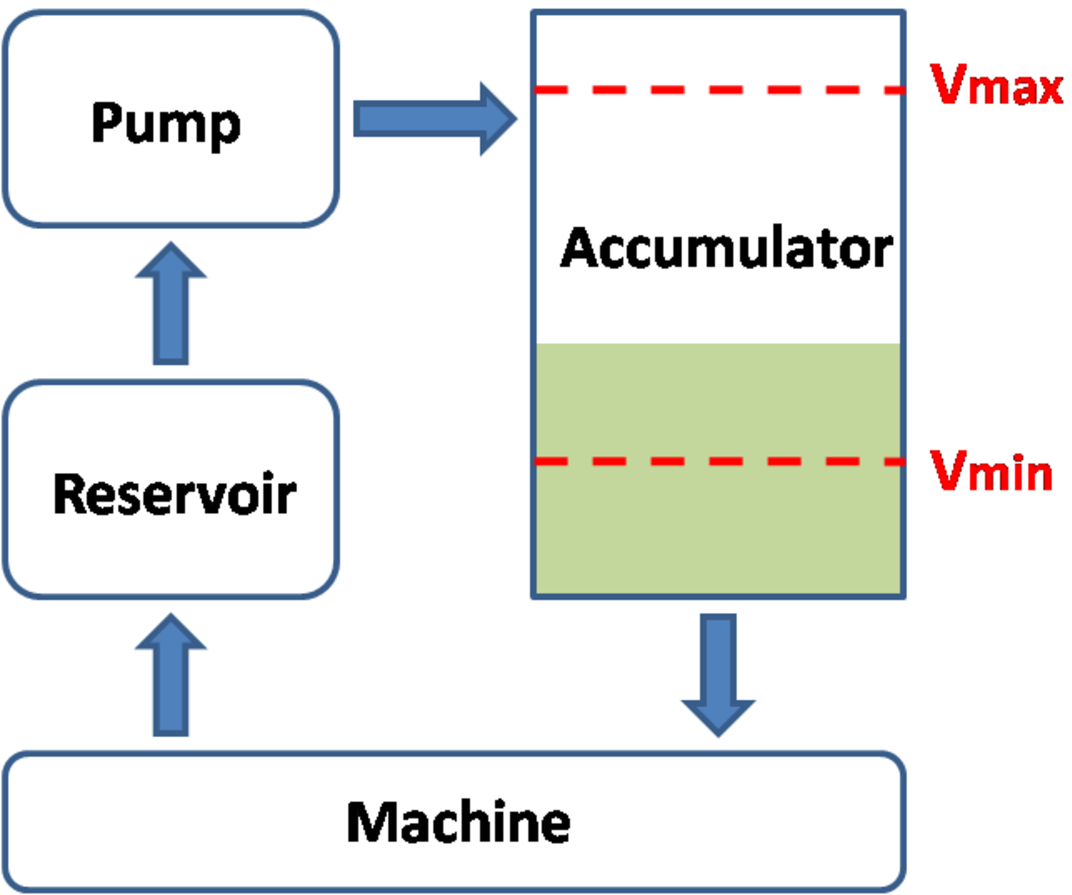}
\caption{The oil pump system. (This picture is based on \cite{kim09}.)}
\label{fig:sys}
\end{center}
\end{minipage}
\hspace{.1cm}
\begin{minipage}[t]{.55\textwidth}
\begin{center}
\includegraphics[width=2.4in,height=1in]{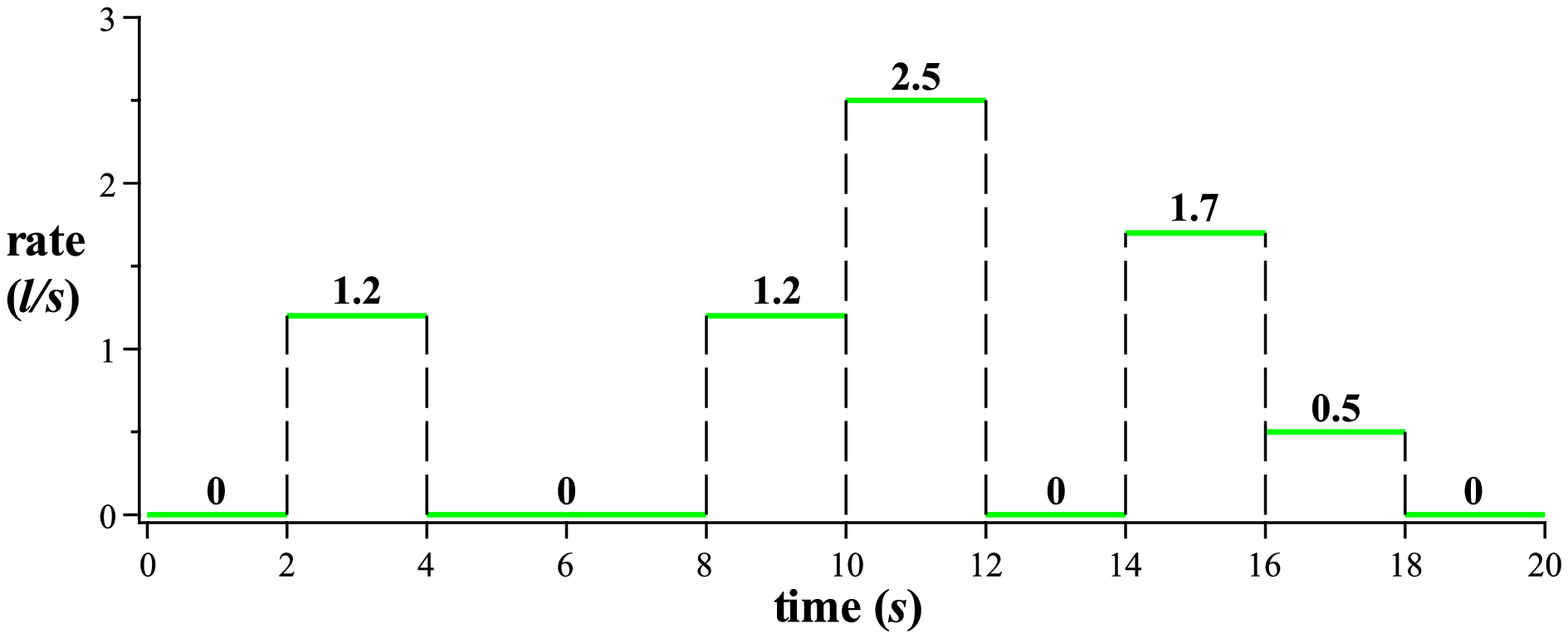}
\caption{Consumption rate of the machine in one cycle.}
\label{fig:consump}
\end{center}
\end{minipage}
\end{figure}

Control objectives for this system are: by switching on/off the pump at certain time points
\begin{equation}\label{eqn:infinite-points}
0\leq t_1\leq t_2\leq \cdots \leq t_n \leq t_{n+1}\leq \cdots\,,
\end{equation}
ensure that
\begin{itemize}
    \item[$\bullet$] $\mathrm{R_s}$\,(\emph{safety}):\, the system can run arbitrarily long while maintaining $v(t)$ within $ [V_{\min},V_{\max}]$ for any time point $t$, where $v(t)$ denotes the oil volume in the accumulator at time $t$, $V_{\min}=4.9\,l$  (liter) and $V_{\max}=25.1\,l$\,;
\end{itemize}
and considering the energy cost and wear of the system, a second objective:
\begin{itemize}
   \item[$\bullet$] $\mathrm{R_o}$\,(\emph{optimality}):\, minimize the average accumulated oil volume in the limit, i.e. minimize $$\lim_{T\rightarrow \infty}\frac{1}{T}\int_{t=0}^T v(t){\mathrm d}t\enspace.$$
\end{itemize}
Both objectives should be achieved under two additional constraints:
\begin{itemize}
   \item[$\bullet$] $\mathrm{R_{pl}}$\,(\emph{pump latency}):\, there must be a latency of at least $2\,s$ between any two consecutive operations of the pump; and
   \item[$\bullet$] $\mathrm{R_r}$\,(\emph{robustness}):\, uncertainty of the system should be taken into account:
       \begin{itemize}
         \item[-] fluctuation of consumption rate (if it is not $0$), up to $f =0.1\,l/s$\,;
         \item[-] imprecision in the measurement of oil volume, up to $\epsilon = 0.06\,l$\,;
         \item[-] imprecision in the measurement of time, up to $\delta =0.015\,s$.\footnote{In \cite{kim09}, $\delta$ is assumed to be $0.01$. Here we include an extra rounding error of $0.005$ due to floating point calculations in the implementation of our control strategy.}
       \end{itemize}
\end{itemize}

In \cite{kim09}, the authors used timed game automata to model the above system, and applied the tool U{\small{PPAAL}}-T{\small{I}}G{\small{A}} to synthesize near-optimal controllers. Due to discretization made in the timed-game model, an incorrect controller might be synthesized. Therefore the correctness and robustness of the synthesized controllers are checked using the tool PHAV{\small{ER}}. Through simulations in S{\small{IMULINK}}, it was shown that the controllers synthesized by U{\small{PPAAL}}-T{\small{I}}G{\small{A}} provides big improvement (about $40\%$) over the {\emph{Bang-Bang Controller}} and {\emph{Smart Controller}} that are currently used at the H{\small{YDAC}} company. We will show how further improvement can be achieved using our approach.

\section{Deriving Constraints from Safety Requirements}
Following \cite{kim09}, the determination of control points (\ref{eqn:infinite-points}) can be localized by exploiting the periodicity of oil consumption. That is, decisions on when to switch on/off the pump in one cycle can be made \emph{locally} by measuring the initial oil volume $v_0$ at the beginning of each cycle. Accordingly, the safety requirement $\mathrm R_{\mathrm{s}}$ in Section \ref{sec:oil} can be reformulated as: find an interval $[L,U]\subseteq [V_{\min},V_{\max}]$ s.t.
\begin{itemize}
    \item[$\bullet$] $\mathrm {R_{lu}}$\,(\emph{constraint for $L,U$}):\, for all $v_0\in[L,U]$, there is a finite sequence of time points
        $\mathbf t=(t_1,t_2,\ldots,t_n)\,,$\footnote{The choice of $n$ will be made later (in this paper $n$ can be $0,2,4,6$), but larger $n$'s obviously will have the potential of allowing improved controllers.}
        where $0\leq t_1\leq t_2\leq \ldots\leq t_n\leq 20$ satisfy $\mathrm {R_{pl}}$, for turning on/off the pump so that the resulting $v(t)$ with $v(0)=v_0$ satisfies
    \begin{itemize}
        \item[$\scriptstyle{\bullet}$] $\mathrm{R_i}$\,(\emph{inductiveness}):\, $v(20)\in [L,U]$; and
        \item[$\scriptstyle{\bullet}$] $\mathrm{R_{ls}}$\,(\emph{local safety}):\, $v(t)\in [V_{\min},V_{\max}]$ for all $t\in [0,20]$
    \end{itemize}
    under constraint $\mathrm {R_r}$.
\end{itemize}
\begin{definition}[Local Controller]\label{dfn:local-ctrl}
The above $\mathbf t$ corresponding to $v_0$ is called a {local controller}; the interval $[L,U]$ is called a stable interval.
\end{definition}

Basically, $\mathrm {R_{lu}}$ says that there is a stable interval $[L,U]$ and a corresponding family of local control strategies which can be repeated for arbitrarily many cycles and guarantee safety in each cycle.

\subsubsection{Modeling  Oil Consumption.}
Let $V_{out}(t)$ with $V_{out}(0)=0$ denote the amount of oil consumed by time $t$ in one cycle, and modify the consumption rate in Fig.~\ref{fig:consump} by $f$ in ($\mathrm{R_r}$). Then by simply integrating the lower and upper bounds of the consumption rate over the time interval $[0,20]$ we can get
\begin{equation*}\label{eqn:consumption}
 C_1\,\define\,\,\,
  \begin{array}{lll}
   \,\,\,\,\scriptstyle(0\leq t\leq 2 & \scriptstyle\longrightarrow & \scriptstyle V_{out}=0)\\
   \scriptstyle\wedge\, (2\leq t\leq 4 &\scriptstyle \longrightarrow & \scriptstyle1.1(t-2)\leq V_{out} \leq
                    1.3(t-2))\\
   \scriptstyle\wedge\, (4\leq t\leq 8 & \scriptstyle\longrightarrow & \scriptstyle 2.2\leq V_{out}\leq 2.6)\\
  \scriptstyle \wedge\, (8\leq t\leq 10 & \scriptstyle\longrightarrow & \scriptstyle 2.2+1.1(t-8)\leq V_{out}
                      \leq 2.6+1.3(t-8))\\
  \scriptstyle \wedge\, (10\leq t\leq 12 & \scriptstyle\longrightarrow & \scriptstyle 4.4+2.4(t-10)\leq V_{out}
                      \leq 5.2+2.6(t-10))\\
   \scriptstyle\wedge\,(12\leq t\leq 14 & \scriptstyle\longrightarrow &\scriptstyle 9.2\leq V_{out}\leq 10.4)\\
  \scriptstyle \wedge\, (14\leq t\leq 16 & \scriptstyle\longrightarrow & \scriptstyle 9.2+1.6(t-14)\leq V_{out}
                      \leq 10.4+1.8(t-14))\\
   \scriptstyle\wedge\,  (16\leq t\leq 18 & \scriptstyle\longrightarrow & \scriptstyle 12.4+0.4(t-16)\leq V_{out}
                      \leq 14+0.6(t-16))\\
   \scriptstyle\wedge\, (18\leq t\leq 20 & \scriptstyle\longrightarrow & \scriptstyle13.2\leq V_{out}\leq 15.2)
  \end{array}
  \enspace .\footnote{In the sequel when a function $\gamma(t)$ appears in a formula, the argument $t$ is dropped and $\gamma$ is taken as a real-valued variable.}
\end{equation*}
Actually, if the machine consuming oil is regarded as a hybrid system $\mathcal H$ with state variable $V_{out}$ and continuous dynamics subject to \emph{box} constraints, then $C_1$ is the exact \emph{reachable set} of $\mathcal H$ from initial point $V_{out}=0$ within 20 time units. Therefore we do not need to approximate the reachable set of $\mathcal H$ by generating inductive invariants. This is also the case with the following pump system. However, if the consumption profile is more complicated, say piecewise polynomial, then approximations are indeed necessary.

\subsubsection{Modeling  Pump.}
In \cite{kim09} it is assumed that the number of activations of pump in one cycle is at most 2. We will adopt this assumption at first and increase this number later on. With this assumption, there will be at most four time points to switch the pump on/off in one cycle, denoted by $0\leq t_1\leq t_2\leq t_3\leq t_4\leq 20$. If the pump is started only once or zero times, then we just set $t_3=t_4=20$ or $t_1=t_2=t_3=t_4=20$ respectively. Then the 2-second latency requirement ($\mathrm{R_{pl}}$) can be modeled by
\begin{equation*}\label{eqn:startimes}
 C_2\,\define\,\,\,
  \begin{array}{l}
   \,\,\,\,\scriptstyle(t_1\geq 2\,\wedge\, t_2-t_1\geq 2\, \wedge\, t_3-t_2\geq 2 \,\wedge\, t_4-t_3\geq 2 \,\wedge\, t_4\leq 20)\\
   \scriptstyle\vee\,  (t_1\geq 2\,\wedge\, t_2-t_1\geq 2 \,\wedge\, t_2\leq 20 \,\wedge\, t_3= 20 \,\wedge\, t_4=20)\\
   \scriptstyle\vee\,  (t_1=20\,\wedge \,t_2=20\,\wedge\, t_3=20\,\wedge\, t_4=20)\\
  \end{array}\enspace .
\end{equation*}

Let $V_{in}(t)$ with $V_{in}(0)=0$ denote the amount of oil introduced into the accumulator by time $t$ in one cycle. Then we have
\begin{equation*}\label{eqn:introduced}
 C_3\,\define\,\,\,
  \begin{array}{lll}
\,\,\,\, \scriptstyle (0\leq t\leq t_1 & \scriptstyle\longrightarrow &\scriptstyle V_{in}=0)\\
   \scriptstyle\wedge\, (t_1\leq t\leq t_2 & \scriptstyle\longrightarrow &\scriptstyle V_{in}=2.2(t-t_1))\\
  \scriptstyle \wedge\,  (t_2\leq t\leq t_3 & \scriptstyle\longrightarrow & \scriptstyle V_{in}=2.2(t_2-t_1))\\
   \scriptstyle\wedge\,  (t_3\leq t\leq t_4 & \scriptstyle\longrightarrow &
                      \scriptstyle  V_{in}=2.2(t_2-t_1)+2.2(t-t_3))\\
  \scriptstyle \wedge\, (t_4\leq t\leq 20 & \scriptstyle\longrightarrow &\scriptstyle  V_{in}=2.2(t_2+t_4-t_1-t_3))\\
  \end{array}
 \enspace .
\end{equation*}

\subsubsection{Encoding  Safety Requirements.}
Denote the oil volume in the accumulator at the beginning of one cycle by $v_0$, and the volume at time $t$ by $v(t)$. Then for any $0\leq t\leq 20$ we have:
\begin{equation*}\label{eqn:acclevel}
 C_4\,\define\,\, v=v_0+V_{in}-V_{out} \enspace.
\end{equation*}

According to ($\mathrm{R_r}$), the measurement of $t_i$ ($1\leq i\leq 4$) and $v_0$ may deviate from their actual values, so $v(t)$ will deviate from its predicted value as stated in the requirement $C_4$. Nevertheless, we have the following estimation of the deviation of $v(t)$.
\begin{lemma}\label{lem:deviation}
Let $\tilde v(t)$ denote the actual oil volume in the accumulator at time $t$. Then for any $0\leq t\leq 20$,
%\begin{equation*}
  $|v(t)-\tilde v(t)|\leq 8.8\,\delta + \epsilon<0.2.$
%\end{equation*}
\end{lemma}
\begin{proof}
By ($\mathrm{R_r}$) and $C_4$, $v_0$ will cause an imprecision of $\epsilon$ and each $t_i$ will cause an imprecision of $2.2\,\delta$\,.
\qed
\end{proof}

By Lemma~\ref{lem:deviation}, it is sufficient to rectify the safety bounds in ($\mathrm{R_i}$) and ($\mathrm{R_{ls}}$) by an amount of $0.2$. Let
\begin{equation*}\label{eqn:sfendpoint}
\begin{array}{lll}
  C_5&\,\define\,&\,\,t=20\longrightarrow L+0.2\leq v\leq U-0.2\,\\
  C_6&\,\define\,&\,\, 0\leq t\leq 20 \longrightarrow V_{\min}+0.2\leq v\leq V_{\max}-0.2 \enspace .
\end{array}
\end{equation*}
Then ($\mathrm{R_i}$) and ($\mathrm{R_{ls}}$) can be expressed as
\begin{equation*}
  \mathcal S\,\define\, \forall t,v,V_{in},V_{out}.(C_1\wedge C_3 \wedge C_4 \longrightarrow C_5\wedge C_6 )\,.
\end{equation*}
\subsubsection{Deriving Constraints.}
To find such $[L,U]$ that for every $v_0\in[L,U]$ there is a local control strategy satisfying $\mathrm {R_{i}}$ and $\mathrm {R_{ls}}$, let
\begin{equation*}\label{eqn:constraintv0}
C_7\,\define\,\, L\leq  v_0\leq U \enspace,
\end{equation*}
and then $\mathrm{R_{lu}}$ can be encoded into
\begin{equation*}\label{eqn:constraintLU}
\begin{array}{l}
C_8\,\define\,\forall v_0. \Big(C_7 \longrightarrow \exists t_1t_2t_3t_4. \big(C_2\,\wedge\mathcal S\big)\Big)
\end{array}\,.
\end{equation*}
%%
%% omitted here, quantifier elimination tool Mjollnir
%% add constraints: L>=4.9,U<=25.1
We use the tool Mjollnir \cite{Monniaux-Mjollnir} to do QE on $C_8$ and the following result is returned:
\begin{equation*}\label{eqn:constraintLU-qf}
  C_9\,\define\, L\geq 5.1 \wedge U\leq 24.9 \wedge U-L\geq 2.4\enspace.
\end{equation*}
%%
%% omitted
%% intuitive meaning of C_8
%%
Then the relation between $L,U,v_0$ and the corresponding local control strategy $\mathbf t=(t_1,t_2,t_3,t_4)$ can be obtained by applying QE to
\begin{equation*}\label{eqn:constraintai}
 C_{10}\,\define\,C_2\wedge C_7\wedge C_9 \,\wedge \mathcal S\,.
\end{equation*}
The result given by Mjollnir, when converted to DNF, is a disjunction of 92 components:
\begin{equation*}\label{eqn:domain7}
\mathcal D(L,U,v_0,t_1,t_2,t_3,t_4)\,\define\,\bigvee_{i=1}^{92} D_i
\end{equation*}
(denoted by $\mathcal D$ for short), with each $D_i$ representing a nonempty closed convex polyhedron (see Appendix \ref{app:dnf92}).\footnote{The fact that each $D_i$ is a nonempty closed set can be checked using QE.}
%% omitted
%% closed convex polyhedra

\section{A ``Hybrid" Approach for Optimization}
\subsection{Encoding of the Optimization Objective}
By Definition \ref{dfn:local-ctrl}, the optimal average accumulated oil volume in $\mathrm {R_{o}}$ can be redefined as
\begin{equation}\label{eqn:opt-criterion}
\bullet\,\,\,\mathrm {R_o':}\quad\quad
\min_{\scriptscriptstyle{[L,U]}}\,\max_{\scriptscriptstyle{v_0\in [L,U]}}\,\min_{\mathbf t} \,\frac{1}{20}\int_{t=0}^{20} v(t)\mathrm{d}t \enspace. \end{equation}
The intuitive meaning of ($\mathrm{R_{o}'}$) is:
\begin{itemize}
\item for each admissible $[L,U]$ and each $v_0\in [L,U]$, minimize the average accumulated oil volume in one cycle, i.e. $\frac{1}{20}\int_{t=0}^{20} v(t)\mathrm{d}t$, over all admissible local controllers $\mathbf t$;
\item fix $[L,U]$ and select the \emph{worst} local minimum by traversing all $v_0\in [L,U]$;
\item then the global minimum is obtained at the interval whose worst local minimum is \emph{minimal}.
\end{itemize}
\begin{definition}[Local Optimal Controller]\label{dfn:local-opt-ctrl}
Let $\mathcal D_{\mathbf t}\,\define\,\{\mathbf t\mid (L,U,v_0,\mathbf t)\in \mathcal D\}$ for fixed $L,U,v_0$. Then we call
$$\min_{\mathbf t\in \mathcal D_{\mathbf t}}\,\frac{1}{20}\int_{t=0}^{20} v(t)\mathrm{d}t$$
the optimal local average accumulated oil volume corresponding to $L,U,v_0$, and the optimizer $\mathbf t$ is called the local optimal controller.
\end{definition}

Let $g(v_0,t_1,t_2,t_3,t_4)\,\define\,\frac{1}{20}\int_{t=0}^{20} v(t)\mathrm{d}t$, denoted by $g$ for short. Then it can be computed from $C_1,C_3,C_4$ without considering fluctuations of consumption rate that
$$ g \,=\, \frac{20v_0+1.1(t_1^2-t_2^2+t_3^2-t_4^2-40t_1+40t_2-40t_3+40t_4)-132.2}{20}\enspace .$$
%  \begin{eqnarray}
%    g&\,\define\,& \frac{1}{20}\int_{t=0}^{20} v(t)\mathrm{d}t \nonumber\\
%    & \,= &
%    \frac{20v_0+1.1(t_1^2-t_2^2+t_3^2-t_4^2-40t_1+40t_2-40t_3+40t_4)-132.2}{20}\,.\nonumber
%  \end{eqnarray}
Then by Proposition \ref{prop:optimization}, ($\mathrm{R_o'}$) can be encoded into
\begin{equation}\label{eqn:objective}
\exists L,U.\Big(C_9\wedge\,\forall v_0.\big(C_7 \longrightarrow \exists t_1t_2t_3t_4.(\mathcal D \wedge g\leq z)\big)\Big)\,,
\end{equation}
which is equivalent to $z\geq z^*$ or $z>z^*$, where $z^*$ equals the value of (\ref{eqn:opt-criterion}).

\subsection{Techniques for Performing QE}\label{sec:techniques}
The above deduced (\ref{eqn:objective}) is a huge formula with nonlinear terms and two alternations of quantifiers, for which direct QE fails. Therefore we have made our efforts to decompose the QE problem into manageable parts.

\subsubsection{Eliminating the Inner Quantifiers.}
We first eliminate the innermost quantified variables $\exists t_1t_2t_3t_4$ by employing the theory of quadratic programming.

Note that $D_i$ in $\mathcal D$ is a closed convex polyhedron for all $i$ and $g$ is a quadratic polynomial function, so minimization of $g$ on $D_i$ is a \emph{quadratic programming} problem. Then the \emph{Karush-Kuhn-Tucker} ({\small{KKT}}) \cite{jensen02} condition
\begin{equation}\label{eqn:quadratic-opt}
\theta_{\mathrm{kkt}}\,\,\define\,\,\,\exists \boldsymbol{\mu}. \,\mathcal L(g,D_i)\,,
\end{equation}
where $\mathcal L(g,D_i)$ is a linear formula constructed from $g$ and $D_i$, and $\boldsymbol{\mu}$ is a vector of new variables, gives a \emph{necessary} condition for a local minimum of $g$ on $D_i$.

By applying the {\small{KKT}} condition to each $D_i$ and eliminating all $\boldsymbol \mu$, we can get a \emph{necessary} condition $\mathcal D'$, a disjunction of 580 parts, for the minimum of $g$ on $\mathcal D$:
\begin{equation*}
  \mathcal D'=\bigvee_{j=1}^{580} B_j\,.
\end{equation*}
Furthermore,
each $B_j$ has the nice property that for any $L,U,v_0$, a \emph{unique} $\mathbf t_j$ is determined by $B_j$ (see Appendix \ref{app:dnf580}).\footnote{This has been verified by QE.}
For instance, one of the $B_j$ reads:
\begin{equation}\label{eqn:afterKKT}
  \begin{array}{l}
  t_4=20 \wedge 16t_2+10L-349=0\, \wedge
  \\ t_2-t_3+2=0 \wedge 22t_1-16t_2-10v_0+107=0\, \wedge \cdots\\
  \end{array}
  \enspace.
\end{equation}
%From (\ref{eqn:afterKKT}), we can see that $t_i$ for $1\leq i\leq 4$ can all be expressed as linear functions in $L,U,v_0$, which actually give
%the local control strategy corresponding to $L,U,v_0$.

Since $\mathcal D'$ keeps the minimal value point of $g$ on $\mathcal D$, the formula obtained by replacing $\mathcal D$ by $\mathcal D'$ in (\ref{eqn:objective})
\begin{equation}\label{eqn:objectiveD'}
\exists L,U.\Big(C_9\wedge\,\forall v_0.\big(C_7 \longrightarrow \exists t_1t_2t_3t_4.(\mathcal D' \wedge g\leq z)\big)\Big)\,
\end{equation}
is equivalent to (\ref{eqn:objective}). Then according to formulas like (\ref{eqn:afterKKT}), $\exists t_1t_2t_3t_4$ in (\ref{eqn:objectiveD'})
can be eliminated by the distribution of $\exists$ among disjunctions, followed by instantiations of $\mathbf t_j$ in each disjunct. Thus (\ref{eqn:objective}) can be converted to
\begin{equation}\label{eqn:objective2}
\exists L,U.\Big(C_9\wedge \forall v_0.\big(C_7 \longrightarrow \bigvee_{j=1}^{580} (A_j\wedge g_j\leq z)\big)\Big)\enspace,
\end{equation}
where $A_j$ is a constraint on $L,U,v_0$, and $g_j$ is the instantiation of $g$ using $\mathbf t_j$ given by formulas like (\ref{eqn:afterKKT}).

\subsubsection{Eliminating the Outer Quantifiers.} We eliminate the outermost quantifiers $\exists L,U$ in (\ref{eqn:objective2}) by discretization, as discussed in Section \ref{sec:overall}.

According to $C_9$, the interval $[5.1,24.9]$ is discretized with a granularity of $0.1$, which gives a set of $199$ elements. Then assignments to $L,U$ from this set satisfying $C_9$ are used to instantiate (\ref{eqn:objective2}). There are totally $15400$ such pairs of $L,U$, e.g. $(5.1,7.5)$, $(5.1,7.6)$ etc, and as many instantiations in the form of
\begin{equation}\label{eqn:objective3}
\forall v_0.\big(C_7 \longrightarrow \bigvee_{j=1}^{580} (A_j\wedge g_j\leq z)\big)\enspace,
\end{equation}
each of which gives an optimal value corresponding to $[L,U]$. In practice, we start from $L=5.1,U=7.5$, and search for the minimal optimal value through all the $15400$ cases with $L$ or $U$ incremented by $0.1$ every iteration.

\subsubsection{Eliminating the Middle Quantifier.}
We finally eliminate the only quantifier left in (\ref{eqn:objective3}) by a divide-and-conquer strategy. First, we can show that
\begin{lemma}\label{lem:equivalence}
  $\bigvee_{j=1}^{580}A_j$ is equivalent to $C_7$ in (\ref{eqn:objective3}).
\end{lemma}
\begin{proof}
As discussed above,
$$\exists \mathbf t. (\mathcal D\wedge g\leq z)\Leftrightarrow \exists \mathbf t. (\mathcal D'\wedge g\leq z)\Leftrightarrow \bigvee_{j=1}^{580} (A_j\wedge g_j\leq z)\enspace.$$
Therefore
$$\exists z\exists\mathbf t. (\mathcal D\wedge g\leq z) \Leftrightarrow \exists z. \bigvee_{j=1}^{580} (A_j\wedge g_j\leq z)\enspace.$$
By eliminating $z$ we have $\exists \mathbf t.\mathcal D\Leftrightarrow \bigvee_{j=1}^{580}A_j$. According to $C_8$ and $C_{10}$, $\mathcal D$ has been chosen in such a way that for any $v_0\in[L,U]$ there is a local controller $\mathbf t$. Thus $\exists \mathbf t.\mathcal D\Leftrightarrow C_7$ when $L,U$ are instantiated.\qed
\end{proof}
By this lemma if all $A_j$ are pairwise disjoint then (\ref{eqn:objective3}) is equivalent to
\begin{equation}\label{eqn:objective4}
\bigwedge_{j=1}^{580}\forall v_0.\big( v_0\in A_j \longrightarrow (A_j\wedge g_j\leq z)\big)\,.
\end{equation}
Since each conjunct in (\ref{eqn:objective4}) is a small formula with only two variables $v_0,z$ and one universal quantifier,  it can be dealt with quite efficiently.
\begin{figure}
\begin{center}
\includegraphics[width=1.5in,height=1.3in]{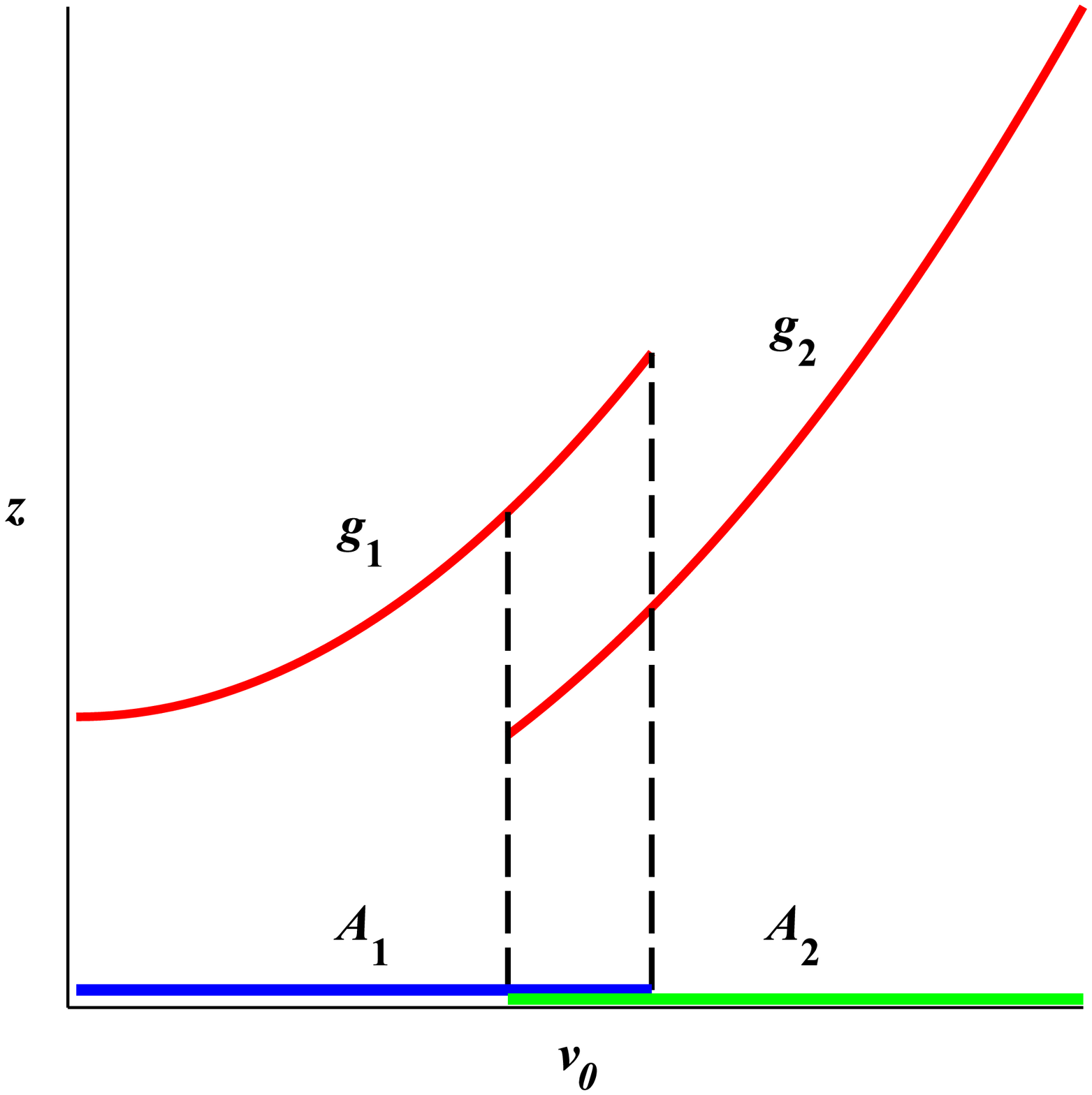}
\hspace{1.5cm}
\includegraphics[width=1.5in,height=1.3in]{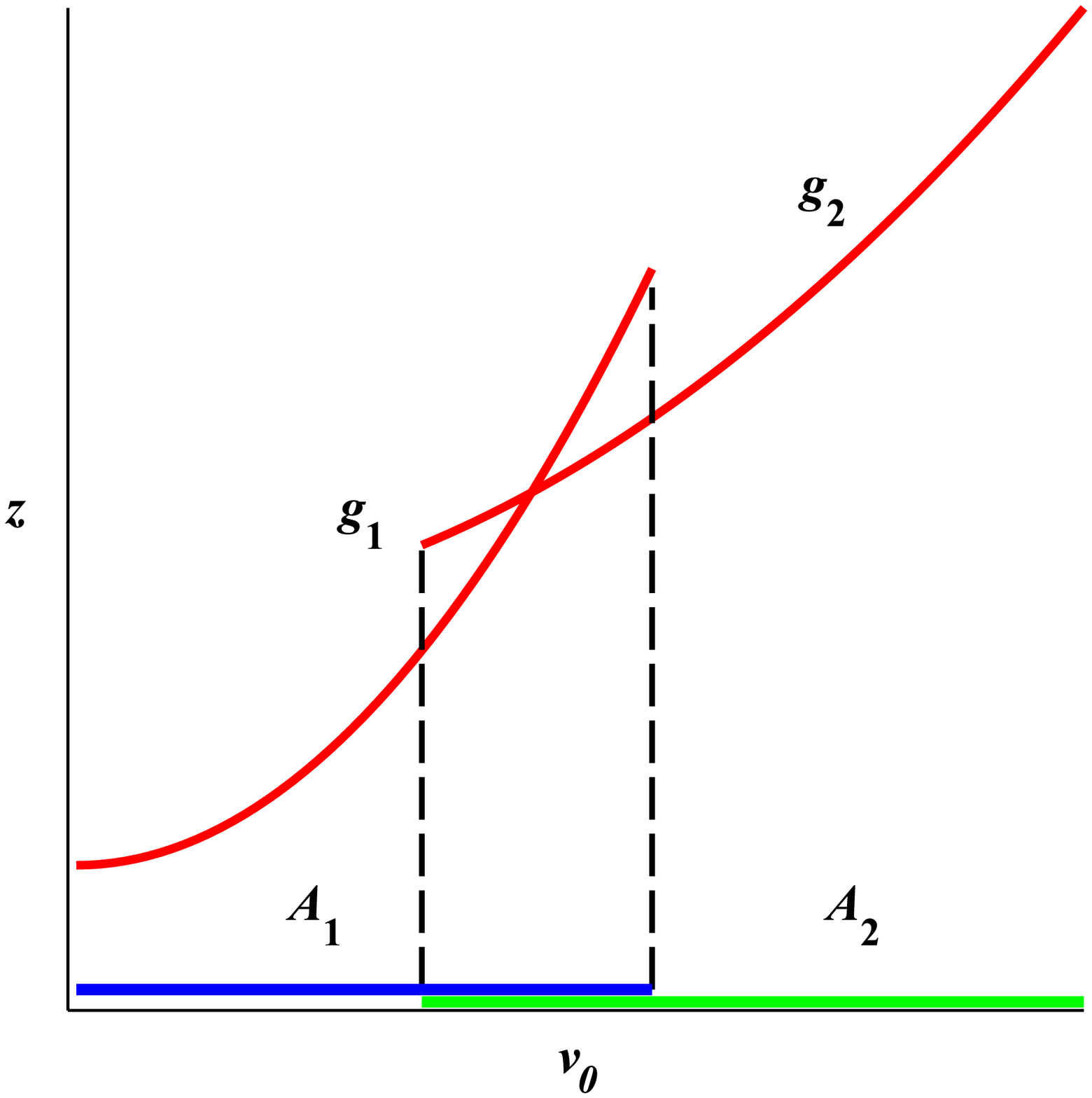}
\caption{Region partition.}
\label{fig:parti}
\end{center}
\end{figure}

If the set of $A_j$s are not pairwise disjoint, then we have to partition them into disjoint regions and assign a new cost function $g_k'$ to each region. The idea for performing such partition is simple and is illustrated by Fig.~\ref{fig:parti}.

Suppose two sets, say $A_1,A_2$, are chosen arbitrarily from the set of $A_j$s. If $A_1\cap A_2=\emptyset$, then we do nothing. Otherwise check wether $g_1\leq g_2$ (or $g_2\leq g_1$) on $A_1\cap A_2$: if so, assign the smaller one, i.e. $g_1\leq z$ (or $g_2\leq z$) to $A_1\cap A_2$; otherwise we simply assign $(g_1\leq z) \vee (g_2\leq z)$ to $A_1\cap A_2$.

If at the same time of partitioning regions we also make a record of the local control strategy in each region, i.e. $\mathbf t_j$, then in the end we can get exactly the family of local optimal controllers corresponding to each $v_0$.

\subsection{Results of QE}
Various tools are available for doing QE. In our implementation, the SMT-based tool Mjollnir \cite{Monniaux08,Monniaux-Mjollnir} is chosen for QE on linear formulas, while R{\small{EDLOG}} \cite{Redlog} implementing \emph{virtual substitution} \cite{Weispfenning93} is chosen for formulas with nonlinear terms. The computer algebra system R{\small{EDUCE}} \cite{Reduce}, of which R{\small{EDLOG}} is an integral part, allows us to perform some programming tasks, e.g. region partition. Table \ref{tbl:timing} shows the performance of our approach. All experiments are done on a desktop running Linux with a 2.66\,GHz CPU and 3\,GB memory.
\begin{table}
\begin{center}
  \caption{Timing of different QE tasks.}\label{tbl:timing}
\begin{tabular}{|c|c|c|c|c|}
  \hline
  formula & $C_8$ & $C_{10}$ & \,$\theta_{\mathrm{kkt}}$ (all 92)\, & all the rest \\
  \hline
  tool & \,Mjollnir\, & \,Mjollnir\, & \,Mjollnir\, &\, Redlog/Reduce \,\\
  \hline
  time & 8m8s & 4m13s & 31s & $<$1s\\
  \hline
\end{tabular}
\end{center}
\end{table}
\paragraph{Remark} In Table \ref{tbl:timing}, timing is in minutes (m) and seconds (s); in the last column, the time taken to get the first optimal value\footnote{For the model with 2 activations, this optimal value is only obtained at the 1st iteration.} is less than 1 second, whereas all 15400 iterations will cost more than 10 hours (using a single computing process).

The final results are as follows:
\begin{itemize}
  \item The interval that produces the optimal value is $[5.1,7.5]$.
  \item The local optimal controller for $v_0\in[5.1,7.5]$ is
   \begin{equation}\label{eqn:oil} t_1=\frac{10v_0-25}{13}\,\wedge\, t_2=\frac{10v_0+1}{13}\, \wedge\, t_3=\frac{10v_0+153}{22}\, \wedge\, t_4=\frac{157}{11}\enspace,
    \end{equation} which is
    illustrated by Picture I in Fig.~\ref{fig:res}. If $v_0=6.5$, then by (\ref{eqn:oil}) the pump should be switched on at $t_1=40/13$, off at $t_2=66/13$, then on at $t_3=109/11$, and finally off at $t_4=157/11$.
  \item The optimal average accumulated oil volume $\frac{215273}{28600}=7.53$ is obtained, improving by $5\%$ the optimal value $7.95$ in \cite{kim09}, which is already a $40\%$ improvement of the controllers from the H{\small{YDAC}} company.
      The local optimal average accumulated oil volume for $v_0\in[5.1,7.5]$ under controller (\ref{eqn:oil}), i.e. $\scriptstyle V_{\textit{aav}}(v_0)\,=\,\frac{1300 v_0^2 + 20420 v_0 + 634817}{114400}$\,, is illustrated by II of Fig.~\ref{fig:res}.
  \item From II of Fig.~\ref{fig:res} we can have an estimate of the performance of controller (\ref{eqn:oil}) in the long run. Without considering noises, it can be computed from (\ref{eqn:oil}) that $v(20)=6.3$ no matter what $v(0)$ is, implying that the mean value of $v_0$ equals 6.3. Therefore the mean average accumulated oil volume in the long run is $V_{\textit{aav}}(6.3)=\frac{40753}{5720}=7.125$. In \cite{kim09}, by simulating the oil pump system for a duration of 200$s$, the mean values $7.44$, $11.56$ and $13.45$ are obtained for the U{\small{PPAAL}}-T{\small{I}}G{\small{A}} controller, Smart Controller and Bang-Bang Controller respectively.
\end{itemize}
\begin{figure}
\begin{center}
  \includegraphics[width=1.7in,height=1.5in]{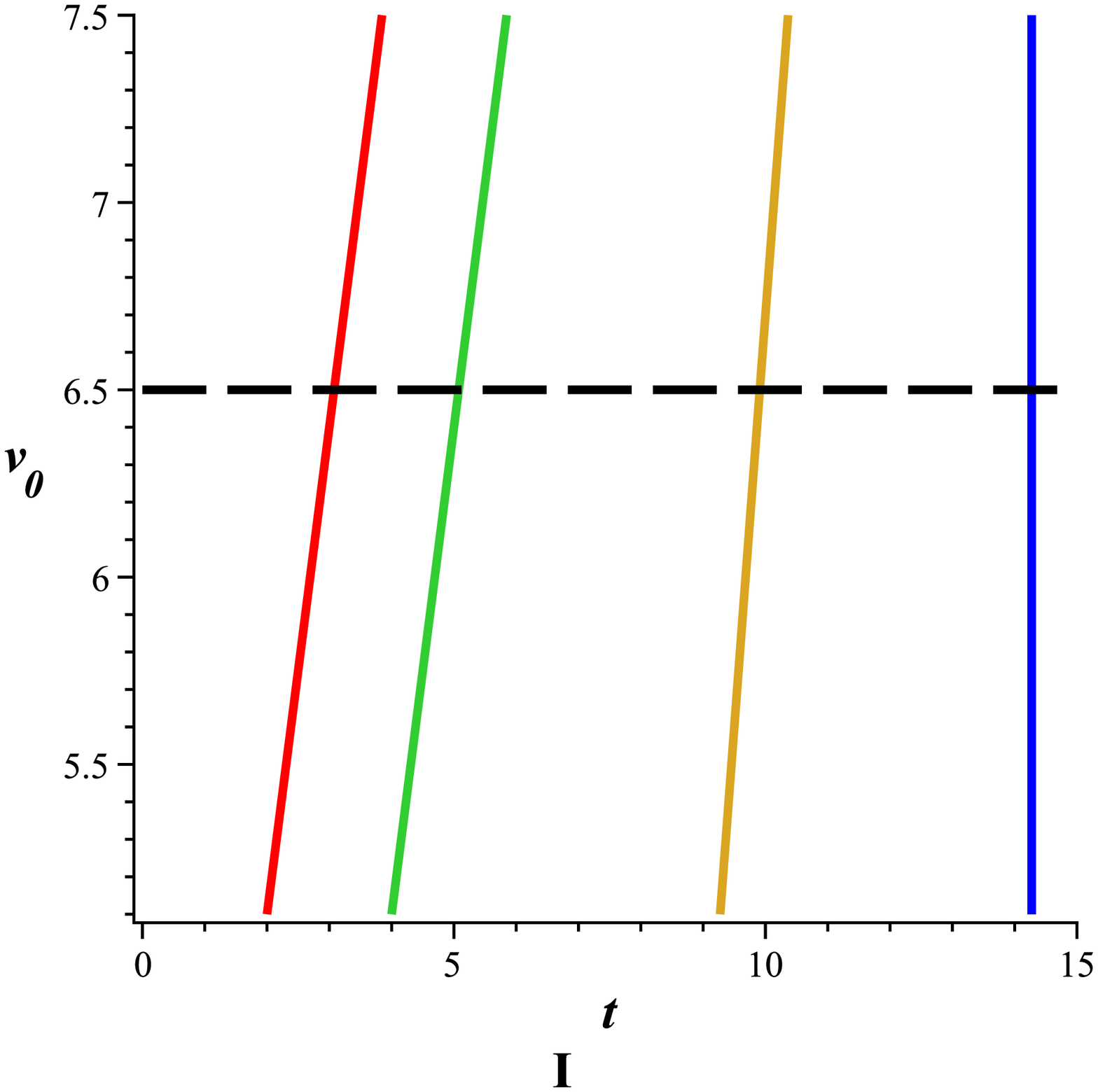}
  \hspace{1cm}
  \includegraphics[width=1.7in,height=1.5in]{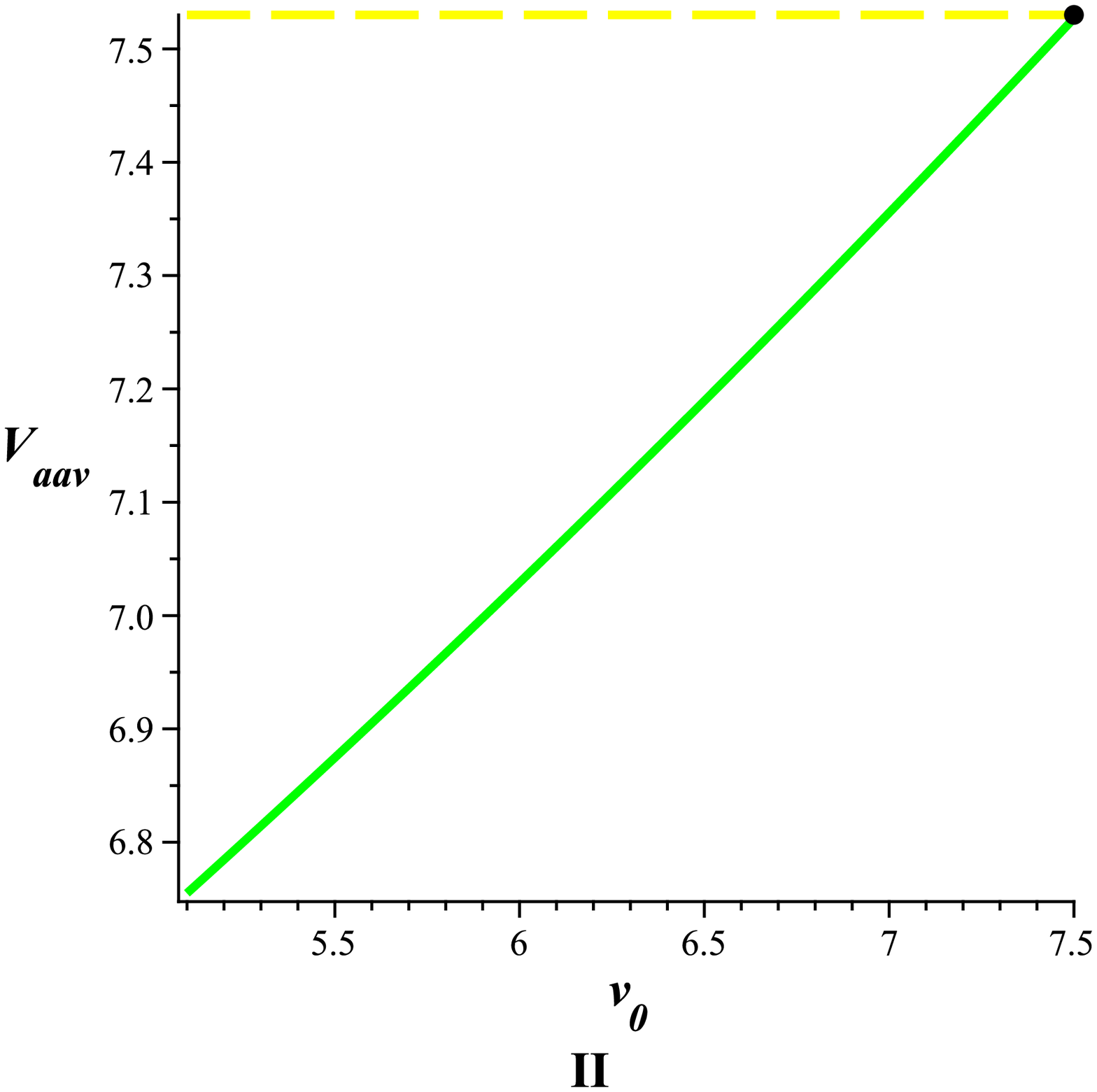}
  \caption{Optimal controllers and average accumulated oil volumes for 2 activations.}
  \label{fig:res}
\end{center}
\end{figure}
%\begin{figure}
%\begin{minipage}[t]{.25\textwidth}
%\begin{center}
%\includegraphics[width=1in,height=1in]{stra.eps}
%%\caption{Optimal controller for 2 activations.}
%\label{fig:strategy}
%\end{center}
%\end{minipage}
%\begin{minipage}[t]{.25\textwidth}
%\begin{center}
%\includegraphics[width=1in,height=1in]{optv.eps}
%%\caption{Optimal value for $v_0\in[5.1,7.5]$.}
%\label{fig:accv0}
%\end{center}
%\end{minipage}
%\begin{minipage}[t]{.25\textwidth}
%\begin{center}
%\includegraphics[width=1in,height=1in]{stra3.eps}
%%\caption{Optimal controller for 3 activations.}
%\label{fig:strategy3}
%\end{center}
%\end{minipage}
%\begin{minipage}[t]{.25\textwidth}
%\begin{center}
%\includegraphics[width=1in,height=1in]{optv3.eps}
%\caption{Optimal value for $v_0\in[5.2,8.1]$.}
%\label{fig:accv03}
%\end{center}
%\end{minipage}
%\end{figure}

\section{Improvement by Increasing Activation Times}
In the controller shown by I of Fig.~\ref{fig:res}, we noticed that when $v_0$ is small and the pump is started on for the second time, it stays on for a period longer than 4 seconds. Based on this observation, we conjecture that if the pump is allowed to be activated three times in one cycle, then each time it could stay on for a shorter period, and the time it is activated for the third time can be postponed. As a result, the accumulated oil volume in one cycle may become less.
%% Comment: The smart controller from HYDAC company also uses 3 siwtches.

To verify the above conjecture, some modifications must be made on the previous model.
Firstly, $C_2$ and $C_3$ should be replaced
\begin{equation*}
 C_2'\,\define\,
  \begin{array}{l}
   \,\,\,\,{\scriptstyle (t_1\geq 2\,\wedge\, t_2-t_1\geq 2\, \wedge \, t_3-t_2\geq 2 \,\wedge \,t_4-t_3\geq 2\, \wedge\, t_5-t_4\geq 2\, \wedge \, t_6-t_5\geq 2\, \wedge\, t_6\leq 20)}\\
   \scriptstyle\vee\,  {\scriptstyle (t_1\geq 2\,\wedge\, t_2-t_1\geq 2 \,\wedge \, t_3-t_2\geq 2\,\wedge\, t_4-t_3\geq 2\,\wedge \,t_4\leq 20 \,\wedge\,  \,t_5= 20 \,\wedge\, t_6=20)}\\
  \scriptstyle \vee\,  {\scriptstyle (t_1\geq 2\,\wedge\, t_2-t_1\geq 2 \,\wedge\, t_2\leq 20\,\wedge\, t_3=20\,\wedge\, t_4=20\,\wedge\, t_5=20\,\wedge \, t_6=20) }\\
   \scriptstyle\vee\,{\scriptstyle (t_1=20\,\wedge\, t_2=20\, \wedge\, t_3=20\,\wedge \, t_4=20\,\wedge \,t_5=20\,\wedge \,t_6=20)}
  \end{array}
\end{equation*}
and
\begin{equation*}
 C_3'\,\define\,\,\,
  \begin{array}{lll}
\,\,\,\,  {\scriptstyle (0\leq t\leq t_1} & \longrightarrow & \scriptstyle{V_{in}=0)}\\
   \scriptstyle\wedge\,  \scriptstyle( t_1\leq t\leq t_2 & \longrightarrow & \scriptstyle V_{in}=2.2(t-t_1))\\
  \scriptstyle\wedge\,  \scriptstyle (t_2\leq t\leq t_3 & \longrightarrow & \scriptstyle V_{in}=2.2(t_2-t_1))\\
   \scriptstyle\wedge\,  \scriptstyle (t_3\leq t\leq t_4 & \longrightarrow &\scriptstyle
                        V_{in}=2.2(t_2-t_1)+2.2(t-t_3))\\
   \scriptstyle\wedge\,  \scriptstyle( t_4\leq t\leq t_5 & \longrightarrow &\scriptstyle
                  V_{in}=2.2(t_2+t_4-t_1-t_3))\\
  \scriptstyle \wedge\,  \scriptstyle( t_5\leq t\leq t_6 & \longrightarrow &\scriptstyle
                  V_{in}=2.2(t_2+t_4-t_1-t_3)+2.2(t-t_5))\\
  \scriptstyle \wedge\,  \scriptstyle( t_6\leq t\leq 20 & \longrightarrow &\scriptstyle
                  V_{in}=2.2(t_2+t_4+t_6-t_1-t_3-t_5))
  \end{array}
\end{equation*}
respectively;
secondly, in $C_5$ and $C_6$ the tolerance of noises should be increased to $0.3$, because due to the increase of times to operate the pump, the maximal uncertainty caused by imprecision in measurement of volume and time is now  $13.2\delta+\epsilon<0.3$; thirdly, the new objective function is
\begin{equation*}
\begin{array}{l}
g\,\define\,{\small{\frac{20v_0+1.1(t_1^2-t_2^2+t_3^2-t_4^2+t_5^2-t_6^2 -40t_1+40t_2-40t_3+40t_4-40t_5+40t_6)-132.2}{20}}}
\end{array}\enspace .
\end{equation*}

%\begin{figure}
%\begin{minipage}[t]{.5\textwidth}
%\begin{center}
%\includegraphics[width=1.8in,height=1.6in]{stra3.eps}
%\caption{Optimal controller for 3 switches.}
%\label{fig:strategy3}
%\end{center}
%\end{minipage}
%\begin{minipage}[t]{.5\textwidth}
%\begin{center}
%\includegraphics[width=1.8in,height=1.6in]{optv3.eps}
%\caption{Optimal value for $v_0\in[5.2,8.1]$.}
%\label{fig:accv03}
%\end{center}
%\end{minipage}
%\end{figure}

For this model, we get the following results.
\begin{itemize}
  \item Using interval $[5.2,8.1]$, the optimal average accumulated oil volume\footnote{The optimal value is first obtained at the 4th iteration, but there are many other intervals other than $[5.2,8.1]$ that give the same optimal value 7.35.} $\frac{6613}{900}=7.35$ is obtained, which is a 7.5\% improvement over the optimum 7.95 in \cite{kim09}.
  \item The local controllers for $v_0\in[5.2,8.1]$ is illustrated
  by I of Fig.~\ref{fig:res3}: %We can see that the controller with three switches is much more complicated than the one with two switches: according to different initial value $v_0$, the global controller is divided into 4 sub-controllers (separated by dashed lines):
  \begin{displaymath}  \left\{
  \begin{array}{ll}
  {\scriptstyle{t_1=\frac{10v_0-26}{13}}\,\wedge \, t_2=\frac{10v_0}{13}\,\wedge \,t_3=\frac{5v_0+76}{11}\, \wedge\, t_4=12\,\wedge\, t_5=14 \,\wedge \,t_6=\frac{359}{22}} & {\scriptstyle v_0\in[5.2,6.8)} \\
  {\scriptstyle{t_1=\frac{10v_0-26}{13}}\,\wedge\, t_2=\frac{10v_0}{13}\,\wedge\, t_3=\frac{5v_0+76}{11} \,\wedge\, t_4=\frac{5v_0+98}{11}\,\wedge \, t_5=\frac{5v_0+92}{9}\,\wedge\,  t_6=\frac{20v_0+3095}{198}}& {\scriptstyle v_0\in[6.8,7.5)}\\
  {\scriptstyle{t_1=\frac{10v_0-26}{13}}\,\wedge \, t_2=\frac{10v_0}{13}\,\wedge\, t_3=\frac{5v_0+76}{11} \,\wedge \, t_4=\frac{5v_0+98}{11}\,\wedge\,  t_5=\frac{5v_0+92}{9}\,\wedge  \, t_6=\frac{5v_0+110}{9}}& {\scriptstyle v_0\in[7.5,7.8)}\\
  {\scriptstyle{t_1=\frac{10v_0+26}{13}}\,\wedge \, t_2=\frac{45v_0+1300}{143}\,\wedge\, t_3=14 \,\wedge\, t_4=\frac{359}{22} \,\wedge\, t_5=20\, \wedge\, t_6=20}& {\scriptstyle{v_0\in[7.8,8.1]}\,.}
  \end{array}\right.
  \end{displaymath}
  \item The local optimal value for $v_0\in[5.2,8.1]$ is illustrated by II of Fig.~\ref{fig:res3}, from which it can be estimated that the mean average accumulated oil volume in the long run is around $6.8$.
\end{itemize}
\begin{figure}
\begin{center}
%  \includegraphics[width=1.7in,height=1.5in]{stra.eps}
%  \hspace{1cm}
%  \includegraphics[width=1.7in,height=1.5in]{optv.eps}
  \includegraphics[width=1.7in,height=1.6in]{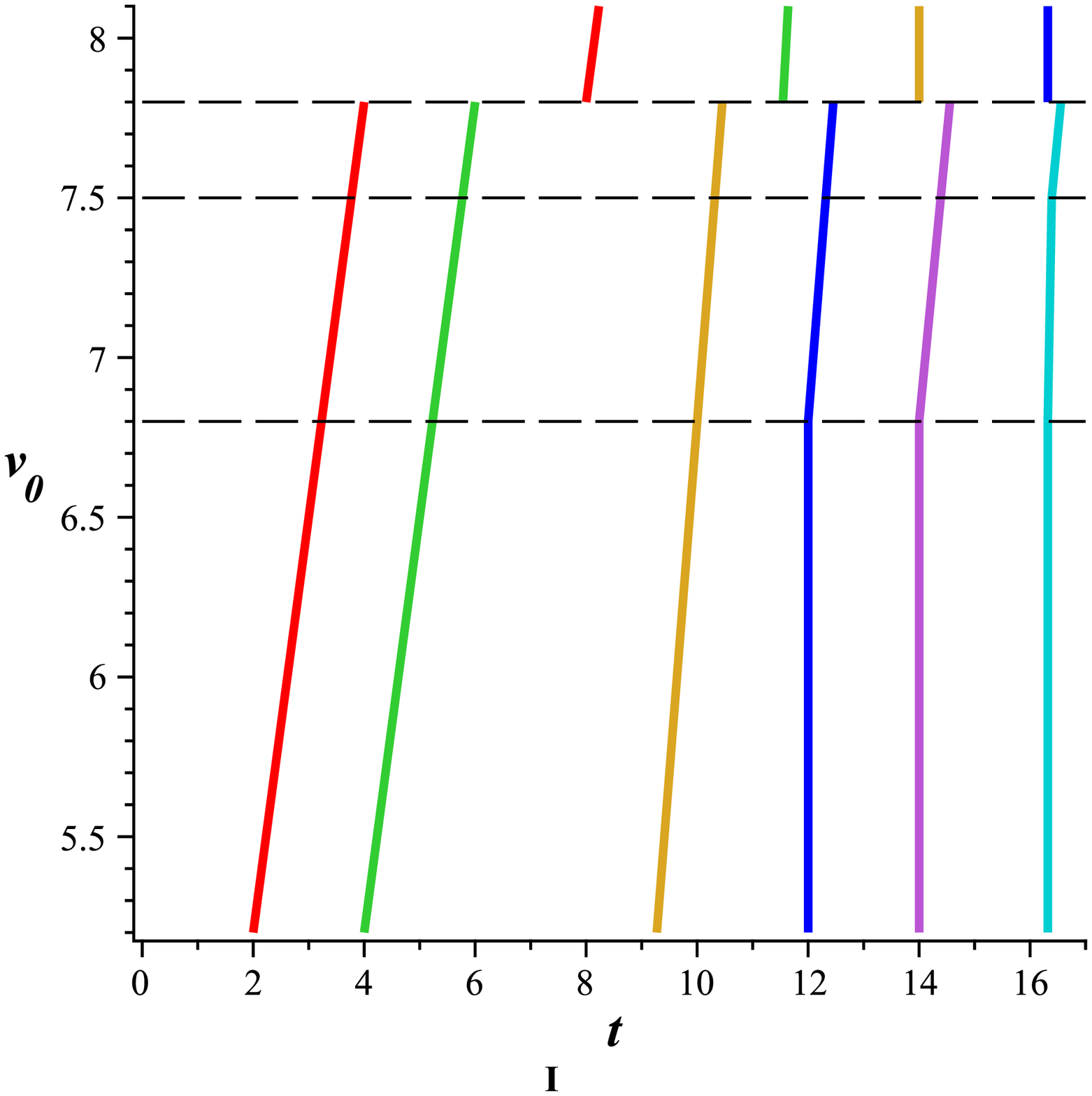}
  \hspace{1cm}
  \includegraphics[width=1.7in,height=1.6in]{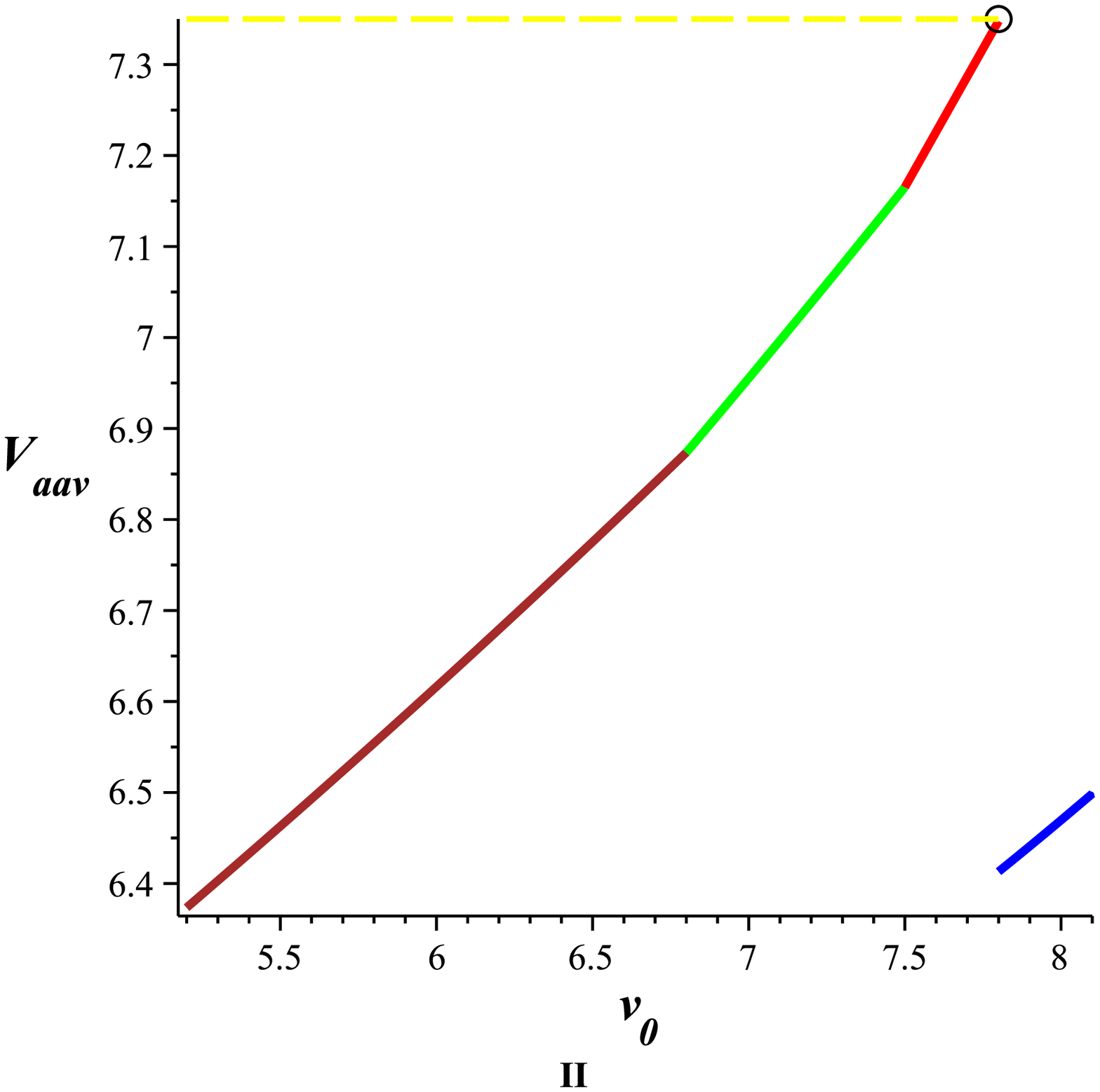}
  \caption{Optimal controllers and average accumulated oil volumes for 3 activations.}
  \label{fig:res3}
\end{center}
\end{figure}

Furthermore, the following theorem indicates that the theoretically optimal controller can be obtained using the local control strategy with three activations. Therefore, our approach in fact gives
the theoretically optimal controller in the oil pump industrial example.
\begin{theorem} \label{th-theoretical-optimal}
For each admissible $[L,U]$, each $v_0\in [L,U]$, and any local control strategy $s_4$ with at least 4 activations subject to $R_{lu}$, $R_i$ and $R_{ls}$, there exists a local control strategy $s_3$ subject to $R_{lu}$, $R_i$ and $R_{ls}$ with 3 activations such that $\,\frac{1}{20}\int_{t=0}^{20} v_{s_3}(t)\mathrm{d}t < \,\frac{1}{20}\int_{t=0}^{20} v_{s_4}(t)\mathrm{d}t$, where $v_{s_3}(t)$ (resp. $v_{s_4}(t)$) is the oil volume in the accumulator at $t$ with $s_3$ (resp. $s_4$).
\end{theorem}
% proof
\begin{proof}
From the consumption rate of the machine in Fig. \ref{fig:consump} and the behavior of the pump, we only need to consider a controller $s_4$ that turns on the pump 4 times in a circle in order to guarantee $R_{lu}$, $R_i$ and $R_{ls}$. Furthermore, by Fig. \ref{fig:consump}, it is easy to argue that turning on the pump can only take place in the intervals $[2,4]$, $[8,12]$ and $[14,20]$ in order to obtain an optimal local control strategy; otherwise, a better local control strategy can be constructed just by postponing the activation time accordingly. In addition, we can further show that the pump can only be turned on at most once in the interval $[8,12]$ for any optimal local control strategy. Now suppose we have an optimal local control strategy $s_4$ that needs to turn on the pump four times in a circle in order to guarantee  $R_{lu}$, $R_i$ and $R_{ls}$. Then by the above analysis, $s_4$ switches the pump on respectively in $[2,4]$, $[8,12]$, at 14 for 2 seconds and at 18 for another 2 seconds. If not, it is easy to show the strategy is not optimal by contradiction. Now, let us construct a local control strategy $s_3$ that turns on the pump three times in a circle as follows: its first two activation time are the same as the counterparts of $s_4$'s, but last $\epsilon$ seconds longer by considering noise, and it turns on the pump the third time at 14 for $3.2 + \epsilon$ seconds, where $\epsilon$ is the noise (0.1 in this paper). By a simple calculation, it is easy to see that $s_3$ satisfies $R_{lu}$, $R_i$ and $R_{ls}$, and $\,\frac{1}{20}\int_{t=0}^{20} v_{s_3}(t)\mathrm{d}t < \,\frac{1}{20}\int_{t=0}^{20} v_{s_4}(t)\mathrm{d}t$. \qed
\end{proof}

\oomit{\section{Related Work}
An adapted version of the oil pump controller case was studied in \cite{Tiwari-emsoft11} using unconstrained nonlinear numerical optimization and learning techniques.

Application of QE in controller synthesis of hybrid systems is not new.
The tool H{\small{Y}}T{\small{ECH}} was the first symbolic model checker that can do parametric analysis \cite{Hytech} for linear hybrid automata, but for the oil pump example it will abort soon due to arithmetic overflow errors. Recently, verification and synthesis of switched dynamical systems using QE were discussed in \cite{sturm-tiwari-issac11}, where the authors gave principles and heuristics for combining different tools, to solve QE problems that are out of the scope of each component tool. }

\section{Conclusions}
In this paper, we propose a ``hybrid" approach for synthesizing optimal controllers of hybrid systems subject to safety requirements by first reducing the problem to QE and then combining symbolic computation and numerical computation for scalability. We illustrate our approach by a real industrial case of an oil pump provided by the H{\small{YDAC}} company.

Compared to the related work, e.g. \cite{kim09}, our approach has the following advantages.
\begin{enumerate}
  \item By modeling the system, safety requirements as well as optimality objectives uniformly and succinctly using first-order real arithmetic formulas, synthesis, verification and optimization are integrated into one elegant framework. The synthesized controllers are guaranteed to be correct.
  \item By combining symbolic computation with numerical computation, we can obtain both high precision and efficiency. For the oil pump example, our approach can synthesize a better (up to $7.5\%$ improvement of \cite{kim09}) optimal controller in a reasonable amount of time (see Table \ref{tbl:timing}), even nearly a theoretically optimal controller by Theorem \ref{th-theoretical-optimal}.
%  \item Our approach does not require decidability properties of hybrid systems, and can therefore be applied to more complex systems beyond timed automata and linear hybrid automata.
\end{enumerate}

The issues of evaluation and implementation of our controllers are being considered. To make our approach more general with symbolic and numerical components, and apply it to more examples in practice will be our future work.

\subsubsection{Acknowledgements.} Special thanks go to Mr. Quan Zhao for his kind help in writing an interface between different QE tools, and to Dr. David Monniaux for his instructions on the use of the tool Mjollnir.

%\bibliographystyle{splncs03}
%\bibliography{fm2012}

\newpage
\appendix
\section{Display of Formulas by QE}
\subsection{The First 10 Disjuncts of $\mathcal D$}\label{app:dnf92}
\begin{displaymath}
\begin{array}{ll}
  & 26 t_1 - 10 v_0 - 157 > 0 \,\wedge\, 22 t_1 - 22 t_2 + 22 t_3 - 22 t_4 - 10 v_0 + 275 < 0 \,\wedge\,\\
  & 11 t_1 - 11 t_2 + 11 t_3 - 11 t_4 - 5 v_0 + 5 U + 65 \geq 0 \,\wedge\,  11 t_1 - 11 t_2 + 11 t_3 - 11 t_4 - 5 v_0 + 5 L + 77 \leq 0 \,\wedge\,\\
  & t_1 - t_2 + 2 \leq 0 \,\wedge\, t_2 - t_3 + 2 \leq 0 \,\wedge\, t_3 - t_4 + 2 \leq 0 \,\wedge\, t_4 - 20 \leq 0 \,\wedge\, 2 v_0 - 31 \geq 0 \,\wedge\,v_0 - U \leq 0 \,\wedge\,\\
  & 10 L - 51 \geq 0\,\wedge \,10 U - 249 \leq 0\\
  \vee
  & 26 t_1 - 10 v_0 - 157 > 0 \,\wedge\, 22 t_1 - 22 t_2 + 22 t_3 - 22 t_4 - 10 v_0 + 275 \geq 0 \,\wedge\, \\
  & 22 t_1 - 22 t_2 + 6 t_3 - 10 v_0 + 95 \leq 0 \,\wedge\, 11 t_1 - 11 t_2 + 11 t_3 - 11 t_4 - 5 v_0 + 5 L + 77 \leq 0 \,\wedge\, \\
  & t_1 - t_2 + 2 \leq 0 \,\wedge\, t_2 - t_3 + 2 \leq 0 \,\wedge\, t_3 - t_4 + 2 \leq 0 \,\wedge\, t_4 - 20 \leq 0 \,\wedge\, 2 v_0 - 31 \geq 0 \,\wedge\, \\
  & v_0 - U \leq 0 \,\wedge\, 10 U - 249 \leq 0 \,\wedge\, 10 L - 51 \geq 0\\
  \vee
  & 26 t_1 - 10 v_0 - 157 > 0 \,\wedge\, 22 t_1 - 22 t_2 + 22 t_3 - 22 t_4 - 10 v_0 + 271 < 0 \,\wedge\, \\
  & 22 t_1 - 22 t_2 + 18 t_3 - 10 v_0 - 97 > 0 \,\wedge\, 11 t_1 - 11 t_2 + 11 t_3 - 11 t_4 - 5 v_0 + 5 L + 77 \leq 0 \,\wedge\, \\
  & t_1 - t_2 + 2 \leq 0 \,\wedge\, t_2 - t_3 + 2 \leq 0 \,\wedge\, t_4 - 20 \leq 0 \,\wedge\, 2 v_0 - 31 \geq 0 \,\wedge\, v_0 - U \leq 0 \,\wedge\,\\
  & 10 U - 249 \leq 0 \,\wedge\, 10 L - 51 \geq 0\\
  \vee
  & 22 t_1 - 11 t_2 - 10 v_0 + 183 \geq 0 \,\wedge\, 22 t_1 - 22 t_2 + 22 t_3 - 11 t_4 - 10 v_0 + 183 \geq 0 \,\wedge\, \\
  & 22 t_1 - 22 t_2 + 22 t_3 - 22 t_4 - 10 v_0 + 341 < 0 \,\wedge\, 13 t_1 - 10 v_0 + 25 \leq 0 \,\wedge\, \\
  & 11 t_1 - 11 t_2 + 11 t_3 - 11 t_4 - 5 v_0 + 5 U + 65 \geq 0 \,\wedge\, 11 t_1 - 11 t_2 + 11 t_3 - 11 t_4 - 5 v_0 + 5 L + 77 \leq 0 \,\wedge\, \\
  & t_1 - t_2 + 2 \leq 0 \,\wedge\, t_1 - 2 \geq 0 \,\wedge\, t_2 - t_3 + 2 \leq 0 \,\wedge\, t_3 - t_4 + 2 \leq 0 \,\wedge\, t_4 - 20 \leq 0 \,\wedge\, \\
  & v_0 - U \leq 0 \,\wedge\, v_0 - L \geq 0 \,\wedge\, 10 U - 249 \leq 0 \,\wedge\, 10 L - 51 \geq 0\\
  \vee
  & 22 t_1 - 6 t_2 - 10 v_0 + 117 < 0 \,\wedge\, 22 t_1 - 11 t_2 - 10 v_0 + 183 \geq 0 \,\wedge\, \\
  & 22 t_1 - 22 t_2 + 22 t_3 + 4 t_4 - 10 v_0 - 157 > 0 \,\wedge\, 22 t_1 - 22 t_2 + 22 t_3 - 22 t_4 - 10 v_0 + 341 < 0 \,\wedge\, \\
  & 11 t_1 - 11 t_2 + 11 t_3 - 11 t_4 - 5 v_0 + 5 U + 65 \geq 0 \,\wedge\, 11 t_1 - 11 t_2 + 11 t_3 - 11 t_4 - 5 v_0 + 5 L + 77 \leq 0 \,\wedge\, \\
  & t_1 - t_2 + 2 \leq 0 \,\wedge\, t_1 - 2 \geq 0 \,\wedge\, t_3 - t_4 + 2 \leq 0 \,\wedge\, t_4 - 20 \leq 0 \,\wedge\, v_0 - U \leq 0 \,\wedge\, v_0 - L \geq 0 \,\wedge\, \\
  & 10 U - 249 \leq 0 \,\wedge\, 10 L - 51 \geq 0\\
  \vee
  & 22 t_1 - 11 t_2 - 10 v_0 + 183 \geq 0 \,\wedge\, 22 t_1 - 22 t_2 + 22 t_3 - 11 t_4 - 10 v_0 + 183 \geq 0 \,\wedge\, \\
  & 22 t_1 - 22 t_2 + 22 t_3 - 22 t_4 - 10 v_0 + 341 < 0 \,\wedge\, 22 t_1 - 10 v_0 + 73 < 0 \,\wedge\, \\
  & 11 t_1 - 11 t_2 + 11 t_3 - 11 t_4 - 5 v_0 + 5 U + 65 \geq 0 \,\wedge\,  11 t_1 - 11 t_2 + 11 t_3 - 11 t_4 - 5 v_0 + 5 L + 77 \leq 0 \,\wedge\, \\
  & t_1 - t_2 + 2 \leq 0 \,\wedge\, t_1 - 2 \geq 0 \,\wedge\, t_4 - 20 \leq 0 \,\wedge\, v_0 - U \leq 0 \,\wedge \, v_0 - L \geq 0 \,\wedge\, \\
  & 10 U - 249 \leq 0 \,\wedge\, 10 L - 51 \geq 0\\
  \vee
  & 22 t_1 - 11 t_2 - 10 v_0 + 183 \geq 0 \,\wedge\, 22 t_1 - 22 t_2 + 22 t_3 - 6 t_4 - 10 v_0 + 117 \geq 0 \,\wedge\, \\
  & 22t_1 - 22 t_2 + 22 t_3 - 11 t_4 - 10 v_0 + 183 < 0 \,\wedge\, 22 t_1 - 22 t_2 + 22 t_3 - 18 t_4 - 10 v_0 + 309 \geq 0 \,\wedge\, \\
  & 11 t_1 - 11 t_2 + 11 t_3 - 11 t_4 - 5 v_0 + 5 U + 65 \geq 0 \,\wedge\, 11 t_1 - 11 t_2 + 11 t_3 - 11 t_4 - 5 v_0 + 5 L + 77 \leq 0 \,\wedge\, \\
  & t_1 - t_2 + 2 \leq 0 \,\wedge\, t_1 - 2 \geq 0 \,\wedge\, t_2 - t_3 + 2 \leq 0 \, \wedge\, v_0 - U \leq 0 \,\wedge\, v_0 - L \geq 0 \,\wedge\, \\
  & 10 U - 249 \leq 0 \,\wedge\, 10 L - 51 \geq 0\\
  \vee
  & 22 t_1 - 11 t_2 - 10 v_0 + 183 \geq 0 \,\wedge\, 22 t_1 - 22 t_2 + 22 t_3 - 22 t_4 - 10 v_0 + 341 \geq 0 \,\wedge\, \\
  & 22 t_1 - 22 t_2 - 10 v_0 + 271 < 0 \,\wedge\, 11 t_1 - 11 t_2 + 11 t_3 - 11 t_4 - 5 v_0 + 5 U + 65 \geq 0 \,\wedge\, \\
  & 11 t_1 - 11 t_2 + 11 t_3 - 11 t_4 - 5 v_0 + 5 L + 77 \leq 0 \,\wedge\, t_1 - t_2 + 2 \leq 0 \,\wedge\,t_2 - t_3 + 2 \leq 0 \,\wedge\, \\
  & t_3 - t_4 + 2 \leq 0 \,\wedge\, t_4 - 20 \leq 0 \,\wedge\, 2 v_0 - 31 \geq 0 \,\wedge\, v_0 - U \leq 0 \,\wedge\, v_0 - L \geq 0 \,\wedge\, \\
  & 10 U - 249 \leq 0 \,\wedge\, 10 L - 51 \geq 0\\
  \vee
  & 22 t_1 - 11 t_2 - 10 v_0 + 183 \geq 0 \,\wedge\, 22 t_1 - 22 t_2 + 22 t_3 - 6 t_4 - 10 v_0 + 117 \geq 0 \,\wedge\, \\
  & 22 t_1 - 22 t_2 + 22 t_3 - 11 t_4 - 10 v_0 + 183 < 0 \,\wedge\, 22 t_1 - 22 t_2 + 22 t_3 - 18 t_4 - 10 v_0 + 309 \geq 0 \,\wedge\, \\
  & 22 t_1 - 22 t_2 - 10 v_0 + 271 < 0 \,\wedge\, 11 t_1 - 11 t_2 + 11 t_3 - 11 t_4 - 5 v_0 + 5 U + 65 \geq 0 \,\wedge\, \\
  & t_1 - 2 \geq 0 \,\wedge\, t_2 - t_3 + 2 \leq 0 \,\wedge\, 2 v_0 - 31 < 0 \,\wedge\, v_0 - L \geq 0 \,\wedge\, 10 U - 249 \leq 0 \,\wedge\, 10 L - 51 \geq0
  \end{array}
\end{displaymath}
\begin{displaymath}
\begin{array}{ll}
  \vee
  & 22 t_1 - 22 t_2 + 18 t_3 - 10 v_0 - 97 > 0 \,\wedge\, 22 t_1 - 22 t_2 + 6 t_3 - 10 v_0 + 95 \leq 0 \,\wedge\, \\
  & 22 t_1 - 10 v_0 - 109 \leq 0 \,\wedge\, 13 t_1 - 10 v_0 - 27 \leq 0 \,\wedge \, 11 t_1 - 11 t_2 + 11 t_3 - 11 t_4 - 5 v_0 + 5 U + 65 \geq 0 \,\wedge\,\\
  & 11 t_1 - 11 t_2 + 11 t_3 - 11 t_4 - 5 v_0 + 5 L + 77 \leq 0 \,\wedge\, t_1 - t_2 + 2 \leq 0 \,\wedge\, t_1 - 2 \geq 0 \,\wedge\, \\
  & t_3 - t_4 + 2 \leq 0 \,\wedge\, t_4 - 20 \leq 0 \,\wedge\, 10 v_0 - 77 \geq 0 \, \wedge\, v_0 - U \leq 0 \,\wedge\, v_0 - L \geq 0 \,\wedge\, \\
  & 10 U - 249 \leq 0 \,\wedge\, 10 L - 51 \geq 0
\end{array}
\end{displaymath}

\subsection{The First 10 Disjuncts of $\mathcal D'$}\label{app:dnf580}
\begin{displaymath}
\begin{array}{ll}
  & t_1 - 14 = 0 \,\wedge\, t_2 - 16 = 0 \,\wedge\, t_3 - 18 = 0 \,\wedge\, t_4 - 20 = 0 \,\wedge\, \\
  & 10 v_0 - 207 = 0 \,\wedge\, 10 U - 207 \geq 0 \,\wedge\, 10 U - 249 \leq 0 \,\wedge\,10 L - 141 = 0\\
  \vee
  & t_1 - 14 = 0 \,\wedge\, t_2 - 16 = 0 \,\wedge\, t_3 - 18 = 0 \,\wedge\, t_4 - 20 = 0 \,\wedge\, \\
  & 10 v_0 - 187 = 0 \,\wedge\, 10 U - 187 \geq 0 \,\wedge\, 10 U - 249 \leq 0 \,\wedge\,10 L - 121 = 0 \\
  \vee
  & t_1 - 14 = 0 \,\wedge\, t_2 - 16 = 0 \,\wedge\, t_3 - 18 = 0 \,\wedge\, t_4 - 20 = 0 \,\wedge\, \\
  & 10 v_0 - 187 > 0 \,\wedge\, 10 v_0 - 207 < 0 \,\wedge\, 5 v_0 - 5 L - 33 = 0 \,\wedge\, v_0 - U \leq 0 \,\wedge\, 10 U - 249 \leq 0 \\
  \vee
  & 22 t_1 - 10 v_0 - 121 = 0 \,\wedge\, t_1 - t_2 + 2 = 0 \,\wedge\,t_1 - t_3 + 4 = 0 \,\wedge\, t_4 - 20 = 0 \,\wedge\, \\
  & 11 t_1 - 138 \geq 0 \,\wedge\, t_1 - 14 < 0 \,\wedge \, 22 t_1 - 10 U - 121 \leq 0 \,\wedge\, 10 U - 249 \leq 0 \,\wedge\, 10 L - 121 = 0 \\
  \vee
  & 11 t_1 - 5 v_0 + 5 L - 121 = 0 \,\wedge\, t_1 - t_2 + 2 = 0 \,\wedge\, t_1 - t_3 + 4 = 0\,\wedge\, t_4 - 20 = 0 \,\wedge\, \\
  & t_1 - 14 < 0 \,\wedge\, 2 v_0 - 31 \geq 0 \,\wedge\, v_0 - U \leq 0 \,\wedge\, 10 U - 249 \leq 0 \,\wedge\, 26 t_1 - 10 v_0 - 157 > 0 \,\wedge\, \\
  & 22 t_1 - 10 v_0 - 121 < 0 \,\wedge\,11 t_1 - 5 v_0 + 5 U - 133 \geq 0 \\
  \vee
  & t_1 - 14 = 0 \,\wedge\, t_2 - 16 = 0 \,\wedge\, t_3 - 18 = 0 \,\wedge\, t_4 - 20 = 0 \,\wedge\, \\
  & 10 v_0 - 187 = 0 \,\wedge\, 10 U - 187 \geq 0 \,\wedge\, 10 U - 249 \leq 0 \,\wedge\,10 L - 51 \geq 0 \,\wedge\, 10 L - 121 < 0 \\
  \vee
  & 22 t_1 - 10 v_0 - 121 = 0 \,\wedge\, t_1 - t_2 + 2 = 0 \,\wedge\, t_1 - t_3 + 4 = 0 \,\wedge\, t_4 - 20 = 0 \,\wedge\, \\
  & 10 U - 249 \leq 0 \,\wedge\, 10 L - 51 \geq 0 \,\wedge\, 10 L - 121< 0 22 t_1 - 10 U - 121 \leq 0\,\wedge\, \\
  & \,\wedge\, 11 t_1 - 138 \geq 0 \,\wedge\, t_1 - 14 < 0\\
  \vee
  & 26 t_1 - 10 v_0 - 157 = 0\,\wedge\, 4 t_1 - 10 L + 85 = 0 \,\wedge\, t_1 - t_2 + 2 = 0 \,\wedge\, t_1 - t_3 + 4 = 0 \,\wedge\, t_4 - 20 = 0 \,\wedge\,\\
  & t_1 - 12 \geq 0 \,\wedge\, t_1 - 14 < 0 \,\wedge\, 10 U - 249 \leq 0\,\wedge\, 26 t_1 - 10 U - 157 \leq 0 \,\wedge\, 4 t_1 - 10 U + 109 \leq 0 \\
  \vee
  & t_1 - 14 = 0 \,\wedge\, t_2 - 16 = 0 \,\wedge\, t_3 - 18 = 0 \,\wedge\, t_4 - 20 = 0 \,\wedge\, \\
  & 10 v_0 - 207 = 0 \,\wedge\, 10 U - 207 \geq 0 \,\wedge\, 10 U - 249 \leq 0 \,\wedge\,10 L - 51 \geq 0 \,\wedge\, 10 L - 141 < 0
  \\
  \vee
  & 11 t_1 - 5 v_0 + 5 U - 133 = 0 \,\wedge\, 11 t_1 - 5 v_0 + 5 L - 121 = 0 \,\wedge\, t_1 - t_2 + 2 = 0 \,\wedge\, t_1 - t_3 + 4 = 0 \,\wedge\,
  \\ & t_4 - 20 = 0\,\wedge\, 2 v_0 - 31 \geq 0 \,\wedge\, 26 t_1 - 10 v_0 - 157 > 0\,\wedge\, 11 t_1 - 133 \leq 0
\end{array}
\end{displaymath}

\end{document}